\documentclass[dvipsnames,letterpaper,11pt]{article}
\usepackage[margin=1in]{geometry}
\usepackage[utf8]{inputenc}
\usepackage{graphicx,fullpage,paralist}
\usepackage{amsmath,amssymb,tabu,amsthm}
\usepackage{dsfont}
\usepackage{bm}
\usepackage{array}
\newcolumntype{C}{>{$}c<{$}}
\usepackage{comment,hyperref}
\usepackage{booktabs}
\usepackage{subcaption}
\usepackage[font={footnotesize}]{caption}
\usepackage{mathtools}
\usepackage{thmtools}
\usepackage{tikz}
\usepackage{tikz-network}
\usepackage{pgfplots}
\usepackage{varwidth}
\usepackage{xstring}
\usepackage[skip=3pt, indent=20pt]{parskip}
\pgfplotsset{compat=1.10}
\usepgfplotslibrary{fillbetween}
\usetikzlibrary{backgrounds}
\usetikzlibrary{patterns}
\usetikzlibrary{arrows}
\usetikzlibrary{calc}
\usetikzlibrary{matrix,decorations.pathreplacing,positioning,fit}
\usepackage{enumitem}
\usepackage{blkarray}
\usepackage{csquotes}
\usepackage{stmaryrd}
\usepackage{natbib}

\usepackage{anyfontsize}
\usepackage{newtxtext,newtxmath}

\definecolor{color0}{RGB}{230,159,0}
\definecolor{color1}{RGB}{86,180,233}
\definecolor{color2}{RGB}{0,158,115}
\definecolor{color3}{RGB}{240,228,66}
\definecolor{color4}{RGB}{0,114,178}
\definecolor{color5}{RGB}{213,94,0}
\definecolor{color6}{RGB}{204,121,167}

\newcommand{\calL}{\mathcal{L}}

\newcommand{\calT}{\mathcal{T}}

\newcommand{\calF}{\mathcal{F}}
\newcommand{\calC}{\mathcal{C}}

\newcommand{\calI}{\mathcal{I}}

\newcommand{\calS}{\mathcal{S}}

\newcommand{\RR}{\mathbb{R}}
\newcommand{\NN}{\mathbb{N}}

\newcommand{\EE}{\mathbf{E}}
\renewcommand{\Pr}{\mathbf{P}}

\newcommand{\R}{\mathbf{R}}

\newcommand{\mycomment}[1]{}

\newtheorem{theorem}{Theorem}
\newtheorem{lemma}{Lemma}

\newtheorem{remark}{Remark}

\newtheorem{proposition}{Proposition}
\newtheorem{claim}{Claim}

\usepackage[textsize=tiny,textwidth=2cm]{todonotes}

\tikzstyle{component}=[draw opacity=0.4,draw=black,line width=1.0cm,line cap=round,line join=round]

\usepackage{multirow}
\usepackage[noend]{algpseudocode}
\usepackage{algorithm,algorithmicx}
\makeatletter
\def\BState{\State\hskip-\ALG@thistlm}
\makeatother

\newlength{\algofontsize}
\hypersetup{
	colorlinks,
	linkcolor={red!50!black},
	citecolor={blue!50!black},
	urlcolor={blue!80!black}
}
\def\equationautorefname~#1\null{(#1)\null}

\usepackage{etoolbox}
\makeatletter
\patchcmd{\hyper@makecurrent}{%
    \ifx\Hy@param\Hy@chapterstring
        \let\Hy@param\Hy@chapapp
    \fi
}{%
    \iftoggle{inappendix}{%true-branch
        \@checkappendixparam{chapter}%
        \@checkappendixparam{section}%
        \@checkappendixparam{subsection}%
        \@checkappendixparam{subsubsection}%
        \@checkappendixparam{paragraph}%
        \@checkappendixparam{subparagraph}%
    }{}%
}{}{\errmessage{failed to patch}}

\newcommand*{\@checkappendixparam}[1]{%
    \def\@checkappendixparamtmp{#1}%
    \ifx\Hy@param\@checkappendixparamtmp
        \let\Hy@param\Hy@appendixstring
    \fi
}
\makeatletter

\newtoggle{inappendix}
\togglefalse{inappendix}

\apptocmd{\appendix}{\toggletrue{inappendix}}{}{\errmessage{failed to patch}}
\begin{document}
\algrenewcommand\algorithmicrequire{\textbf{Input:}}
\algrenewcommand\algorithmicensure{\textbf{Output:}}

\title{Optimal Guarantees for Online Selection over Time
\vspace{.4cm}
}

\author{
Sebastian Perez-Salazar
\thanks{Department of Computational Applied Mathematics and Operations Research, Rice University, USA.}
\thanks{Ken Kennedy Institute, Rice University, USA.}
\and Victor Verdugo
\thanks{Institute for Mathematical and Computational Engineering, Pontificia Universidad Católica de Chile, Chile.}
\thanks{Department of Industrial and Systems Engineering, Pontificia Universidad Católica de Chile, Chile.}
}

\date{}
\maketitle

\begin{abstract}

Prophet inequalities are a cornerstone in optimal stopping and online decision-making. Traditionally, they involve the sequential observation of $n$ non-negative independent random variables and face irrevocable accept-or-reject choices. The goal is to provide policies that provide a good competitive ratio against the optimal offline solution that can access all the values upfront---the so-called prophet value. In the \emph{prophet inequality over time problem} (POT), the decision-maker can commit to an accepted value for $\tau$ units of time, during which no new values can be accepted. This creates a trade-off between the duration of commitment and the opportunity to capture potentially higher future values. 

In this work, we provide best possible worst-case competitive ratios in the IID setting of POT for single-threshold algorithms and the optimal dynamic programming policy. We show a single-threshold algorithm that achieves a competitive ratio of $(1+e^{-2})/2\approx 0.567$, and we prove that no single-threshold algorithm can surpass this guarantee. 
With our techniques, we can analyze simple algorithms using $k$ thresholds and show that with $k=3$ it is possible to get a competitive ratio larger than $\approx 0.602$. Then, for each $n$, we prove it is possible to compute the tight worst-case competitive ratio of the optimal dynamic programming policy for instances with $n$ values by solving a convex optimization program. A limit analysis of the first-order optimality conditions yields a nonlinear differential equation showing that the optimal policy's asymptotic worst-case competitive ratio is $\approx 0.609$. 
\end{abstract}
%{\small  
%\begin{itemize}
%\item Remove RO: {\color{red} done}
%\item Redo Section 1 (explain model better (Marco), provide more details of proofs)
%\item Redo Section 2 (Move notation to where it is used, present results in a high level or move them to Section 1)
%\item Redo Section 3 (Same as section 2, move complex proofs to appendix, move ODE to main (R2))
%\end{itemize}}

\thispagestyle{empty}
\newpage
\setcounter{page}{2}

\section{Introduction}\label{sec:intro}

%{\color{red} Change order of observation $X_n ,\ldots, X_1$}
Prophet inequalities~\citep{hill1982comparisons,kertz1986stop,krengel1977semiamarts,samuel1984comparison} have become a fundamental model for studying Bayesian problems in the last fifty years. 
In the classic prophet inequality formulation, a sequence of $n$ non-negative independent random variables $X_n,\ldots,X_1$ with known distributional information is revealed one by one to a decision-maker. Upon observing $X_t$ and unaware of future values, the decision-maker has to irrevocably accept the value $X_t$ and stop the process or disregard the value forever to observe the value $t-1$, if any. Hence, the decision-maker faces the dilemma of accepting a value that seems favorable versus the opportunity to observe a better value in the future. The decision-maker is interested in finding an algorithm (a.k.a., policy or strategy) to maximize her expected accepted value. To measure the quality guarantee of such an algorithm, we use the \emph{competitive ratio}, which is the ratio between the expected value obtained by the algorithm and the expected offline maximum $\EE[\max\{X_n,\ldots,X_1 \}]$---the so-called \emph{prophet value}. This ratio can be interpreted as a measure of the \emph{price paid} by a decision-maker who cannot observe all values at once. It is known that a simple single threshold rule produces a strategy where the decision-maker can obtain a competitive ratio of $1/2$, and this is best possible~\citep{krengel1977semiamarts,samuel1984comparison}. 

{In this work, we focus on \emph{prophet inequality over time} (POT) models~\citep{abels2023prophet, disser2020hiring}. As opposed to the classic prophet inequality problem, the decision-maker can accept a value for multiple units of time; however, this precludes accepting any other values during the interval covered by the accepted value. Specifically, upon observing the $t$-th value $X_t$, the decision-maker chooses an integer $1 \leq \tau \leq t$, receives a value of $\tau \cdot X_t$, but loses the opportunity to select values observed in the interval $[t-1,t-\tau+1]$. This dynamic poses a new dilemma for the decision-maker: if she accepts a value for too long, she risks missing higher values observed during the committed time, whereas if the decision-maker accepts a value for a short time, this value might have been the largest in the remainder of the time. To benchmark an algorithm in POT, the competitive ratio corresponds to the expected value obtained by the algorithm and the expected optimal offline solution $\sum_{t=1}^n \EE[\max\{X_n,\ldots,X_{n-t+1} \}]$. The POT models has been studied in both cost-minimization version~\citep{disser2020hiring} and in its profit-maximization version~\citep{abels2023prophet,cristi2024planningprophet}, with both versions having implications for procurement problems (see, e.g.,~\citep{aminian2023real,qin2023minimization}).}

\cite{abels2023prophet} focus on the independent and identically distributed (IID) setting, where all the observed values share a common distribution. The authors characterize the optimal algorithm as the one that computes a sequence of decreasing thresholds such that if the $t$-th observed value $X_t$ is larger than the $t$-th threshold, the value $X_t$ is accepted for the remainder $t$ units of time. 
%Otherwise, $X_t$ is accepted for only one unit of time. 
They provide a single threshold algorithm with a competitive ratio of $\approx 0.396$ and an algorithm with several decreasing thresholds that attain an asymptotic competitive ratio of $0.598$ when $n$ goes to infinity. The authors also provide an upper bound on the worst-case competitive ratio of the optimal algorithm, given by the inverse of the golden ratio $\varphi^{-1} = 2/(1+\sqrt{5})\approx 0.618$. 
In recent work,~\cite{cristi2024planningprophet} provide an algorithm with a competitive ratio of $0.5$, which also works for instances with independent but not necessarily identically distributed random variables.

\subsection{Our Contribution and Results}

%{\color{red} Changed optimal dynamic programming policy by optimal algorithm. Marco complained about it. Just want to move DP when we use it.}
In this work, we provide best possible worst-case competitive ratios in the IID setting for single-threshold algorithms and the optimal algorithm. 
In what follows, we summarize our results and techniques according to the nature of the algorithm and model variant: algorithms with few thresholds, and optimal algorithm.
%\paragraph{Bayesian setting} Our first result is an improved guarantee for single threshold algorithms. In short, this class of algorithms compute a value $x_t$ based on the CDF $F$ and at time $t$ if $X_t\geq x_t$, then $X_t$ is accepted for the remainder $n-t+1$ units of time, while if $X_t< x_t$, then, $X_t$ is accepted only for one unit of time. 

\noindent{\bf Improved guarantees via simple algorithms.} We first focus on algorithms that use a limited number of thresholds. This class of algorithms splits the integer interval $[1,n]$ into $k$ consecutive intervals $I_k,\ldots,I_1$ and in each interval assigns a threshold $z_k,\ldots,z_1>0$, respectively. Then, if $X_t\geq z_i$ and $t\in I_i$, the value $X_t$ is accepted for the remaining $t$ units of time; otherwise, $X_t$ is accepted only for one unit of time. 
In our first contribution, we show a single-threshold algorithm (i.e., $k=1$) achieving a competitive ratio of {$(1 + e^{-2})/2\approx 0.567$ for any $n\geq 1$}.
Furthermore, we show this guarantee is tight, as for every $\varepsilon>0$, there exists an instance for which the competitive ratio of any single-threshold algorithm can not be larger than $(1 + e^{-2})/2+\varepsilon$.
%This lower bound exhibits a fast convergence to $(1+e^{-2})/2$, and in fact, for small values of $n$ is strictly larger than the worst case bound of $(1 + e^{-2})/2$.
In particular, our competitive ratio improves upon the $0.396$ previously single-threshold best-known guarantee provided by~\cite{abels2023prophet}. 
%The tightness result follows by finding an explicit hard distribution for single-threshold algorithms. 

Our single-threshold algorithm is quantile-based: it receives a quantile $q\in (0,1)$ and computes the threshold $z$ such that $q=\Pr(X\geq z)$. We prove an instance-independent lower bound on the competitive ratio of any algorithm that uses a quantile $q$ and use this bound to show that the quantile $q=2/(n+1)$ produces the desired competitive ratio. 
We remark that a similar quantile was computed by~\cite{abels2023prophet}, but our general quantile lower bound provides a tighter analysis. 
In a nutshell, we stray from the standard approximate stochastic dominance technique used in prophet inequalities, and instead, we use a functional density argument to compute the value of the algorithm and the optimal offline value in terms of the derivative of the inverse of the probability distribution. This permits a sharper ratio comparison between the value of the algorithm and the optimal offline value. 

Our methodology is highly generalizable for larger number of thresholds. For any $k$, we prove a general instance-independent lower bound for the competitive ratio obtained by algorithms that use $k$ thresholds $z_k,\ldots,z_1$ that are computed via $q_t=\Pr(X \geq z_t)$. We utilize this formula to provide asymptotic competitive ratios for a small number of thresholds: for $k=2$ thresholds, we obtain a competitive ratio of at least $0.596$, and for $k=3$ thresholds, we obtain a competitive ratio of at least $0.602$. In particular, this shows that three thresholds are sufficient to surpass the best current asymptotic competitive ratio of $0.598$ by~\cite{abels2023prophet}. Using this methodology, we are able to compute numerically lower bounds on the asymptotic competitive ratio for algorithms that use up to $k=6$. All the proofs appear in Section~\ref{sec:small_thresholds}.

\noindent{\bf Optimal policy guarantees via convex optimization.} %{\color{red} mention here that the optimal solution can be found via a DP}
Next, we study the competitive ratio of the optimal policy.
We remark that, as shown by ~\cite{abels2023prophet}, the optimal policy can be obtained using dynamic programming (see Section \ref{sec:prelim} for more details).
We show that for each $n$, the optimal policy's worst-case competitive ratio $\gamma_n$ can be obtained by solving a convex optimization program over a simple polyhedron in the positive orthant. In a nutshell, the variables of this program represent the difference between consecutive optimal thresholds in a worst-case instance, and the constraints encode the monotonicity requirements to have thresholds achievable by the optimal policy. On the other hand, the objective captures the difference between the optimal policy value, scaled up by a factor of $1+\varepsilon$, and the benchmark; we show this function is, in fact, convex in the feasible region. We remark that finding the smallest $\varepsilon$ for which the optimal value is non-negative is equivalent to finding the worst-case competitive ratio, which is recovered by $1/(1+\varepsilon)$.

Our approach is based on two key steps.
In the first step, we reduce the problem of finding the worst-case competitive ratio to an infinite-dimensional optimization problem over the set of positive sequences satisfying three particular monotonicity properties.
We show that these monotonicity properties capture precisely the space of possible threshold sequences defined by the optimal policy.
%To prove this equivalence, we show that the space of worst-case distributions can be fully characterized by the threshold sequences $(\tau_n(F))_{n\in \NN}$ in the optimal dynamic programming policy, and furthermore, the set of threshold sequences via the optimal dynamic programming policy is precisely the space of sequences satisfying \ref{mon1}-\ref{mon3}. 
While the first step already allows us to reduce the problem to a constrained problem in infinite dimension, in the second step, we study the structure of the infinite-dimensional optimization problem and show that it can be further reduced to a nicely behaved finite-dimensional convex optimization problem.
We provide the analysis in Section \ref{sec:cvx-optimal}.

Our characterization theorem allows us to recover the tight worst-case competitive ratio of the optimal policy for each value of $n$, which is equal to $1/(1+\varepsilon_n)$, where $\varepsilon_n$ is an optimal parameter associated with the convex program.
Using the system of first-order optimality conditions, we can compute explicitly $\varepsilon_n$, and by performing a limit analysis of this system when $n\to \infty$, we obtain a differential equation where the optimal asymptotic competitive ratio is embedded as a parameter.
More specifically, we seek a function $y:[0,1]\to [0,1]$ such that
%\begin{align*}
%    h(u(x))' + \int_x^1 \frac{1}{s} h(u(s))'\, \mathrm{d}s - \left( \frac{1-e^{-u(x)}}{u(x)}\right) - u(x)h(u(x)) &= \varepsilon,\\
%    u(0) & = +\infty \\
%    u(1) & = 0.
%\end{align*}
\begin{align}
    h(-\ln (y(t)))' &= (1+\varepsilon)\cdot  t\cdot  \exp \left({\int_t^1 \ln (y(s))\, \mathrm{d}s}\right) \text{ for every } t\in (0,1), \label{ode:1}\\
    y(0) = 0,\; y(1) & = 1, \label{ode:3}
\end{align}
where $h(u)=  (1-e^{-u}(1+u))/u^2$ and $\varepsilon$ is the limit of $\varepsilon_n$. 
We provide more details of this asymptotic analysis in Section~\ref{app:asymptotic-cvx}. 
{  By solving numerically the System~\eqref{ode:1}-\eqref{ode:3}, we obtain $\varepsilon \approx 0.64088$ which shows that the optimal policy has an asymptotic competitive ratio of $\lim_n 1/(1+\varepsilon_n)\approx 0.6094$.}

In Table~\ref{tab:table_1}, we summarize our improved guarantees and comparison with existing results.
\begin{table}[h!]
        %\footnotesize
        \centering
        \resizebox{\columnwidth}{!}{\begin{tabular}{|c|c|c|c|c|}
        \hline
             Type of result       & \multicolumn{2}{c|}{Lower bounds} & Upper bounds \\
            \hline\hline
             &  $k$ thresholds & Optimal policy &  \\ \cline{2-3}
%            &  $k=1$ & $k=2$ &$k=3$ & $k=n$ & \\
            %\hline
             & {$(1+e^{-2})/2$ ($k=1$, any $n$)}   &  $1/(1+\varepsilon_n)$ (opt., any $n$)   & $1/(1+\varepsilon_n)$ (opt., any $n$) \\
             Our results & 0.587 $(k=2,n\to \infty)$  & $0.609$ (opt. $n\to \infty$)  &  $0.609$ (opt. $n\to \infty$)  \\
             & 0.602 $(k=3,n\to \infty)$   &   & $(1+e^{-2})/2 $ ($k=1$ threshold) \\
            \hline
            \cite{abels2023prophet}  & 0.398 ($k=1$, any $n$) & $0.598$ ($n\to \infty$) & $(\sqrt{5}-1)/2\approx 0.618$ \\
            \hline
            \cite{cristi2024planningprophet}  & - & $0.5$ (any $n$, non-IID) &  - \\
            \hline
        \end{tabular}}
        \caption{{Known competitive ratios for POT. }}
        \label{tab:table_1}
    \end{table}

\subsection{Related Work}

The prophet inequality problem was introduced by Krengel, Sucheston, and Garling~\citep{krengel1977semiamarts}. In the last decade, the prophet inequality problems have gained increasing attention for their applicability in mechanism design and pricing; see, e.g.,~\cite{Alaei2011,Chawla2010,Duetting2020,Hajiaghayi2007,Kleinberg2012,Correa2023combinatorial}. The IID prophet inequality introduced by~\cite{hill1982comparisons} is a special case interesting in its own right.~\cite{hill1982comparisons} initially proved that a competitive ratio of $1-1/e$ was possible with a single threshold algorithm while an upper bound of $1/\beta^* \approx 0.745$ was proven, where $\beta^*$ is the unique parameter such that the ordinary differential equation
$y'(t) = y(t)(\ln y(t)-1) - (\beta-1)$ with $y(0) = 1,$
has a solution $y:[0,1]\to [0,1]$ such that $y(1)=1$~\citep{kertz1986stop}. It was recently proved that an algorithm with a competitive ratio $1/\beta^*$ exists for the IID prophet inequality~\citep{correa2021posted}. 
%For the model over time with general independent distributions, recently,~\cite{cristi2024planningprophet} showed that a competitive ratio of $1/2$ is possible.

%Central to our work is the study of threshold-based algorithms (see Section~\ref{sec:small_thresholds}). 
In the original work by~\cite{abels2023prophet}, the authors provide the first $0.396$ competitive ratio with an algorithm using one threshold. Single-threshold algorithms are ubiquitous in prophet inequalities, and 
in fact, the solution presented by~\cite{samuel1984comparison} for the classic prophet inequality uses a single threshold. At the same time, the first guarantee of $1-1/e$ for the IID prophet inequality problem by~\cite{hill1982comparisons} is also based on using a single threshold. Single-threshold algorithms play a major role in optimal stopping problems due to their simplicity, interpretability, and connection to posted price mechanisms~\citep{arnosti2023tight,chawla2023static,correa2019pricing}. In recent years, there has been a growing interest in understanding the value of using a larger number of thresholds in different problems as a way to interpolate between the single-threshold and the optimal dynamic policy~\citep{hoefer2023threshold,hoefer2024stochastic,perez2022iid}. 

Random order models provide a midpoint between Bayesian and adversarial settings. Arguably, the most well-known random order problem is the secretary problem introduced by~\cite{gilbert1966recognizing}; see~\citep{freeman1983secretary} for a classic survey on secretary problems. Random order models have been extensively studied in several online selection problems, including matchings~\citep{bernstein2023improved}, knapsacks~\citep{kesselheim2014primal,albers2021improved}, and matroids~\citep{babaioff2018matroid,soto2021strong,feldman2018simple,berczi2025matroid}; we refer to~\cite{gupta2021random} for a recent survey in random order models. A related problem with the POT is the \emph{temp secretary problem}~\citep{fiat2015temp,kesselheim2017temp} in which multiple selections can be made, and each selection lasts some fixed amount of time. Unlike this model, our selection can be arbitrarily short or large.

%In our work, we make use of several mathematical programs to obtain tight bounds on different classes of algorithms. Our techniques closely resemble the work~\citep{perez2022iid}, with the main difference that we use the derivative of $F^{-1}$, the inverse of the CDF, to provide sharp ratio comparisons against the prophet. Our work also resemble the work~\citep{brustle2024competition}; however, the structure of the optimal solution and the optimal policy make the analysis completely different, requiring new analyses.

%\spcom{Explain the convex program here.}

{\section{Preliminaries}\label{sec:prelim}
The set $\NN$ denotes the non-negative integers $\{0,1,\ldots\}$.
We denote by $\calF$ the set of distributions $F$ over the non-negative reals, with finite positive expectation, and such that $\omega_0(F)<\omega_1(F)$, where $\omega_0(F)=\inf\{y:F(y)>0\}$ and $\omega_1(F)=\sup\{y:F(y)<1\}$ are the left and right endpoints of the support of $F$. 
Given an algorithm $\pi$ for the POT problem, we denote by $v_n(\pi,F)=\EE\left[  \sum_{t=1}^n \tau_t X_t \right]$ the expected value of the algorithm, where $\tau_t\geq 0$ corresponds to the accepted time given to $X_t$, with $\tau_t=0$ if $t$ was covered by previous times. Then, the \emph{competitive ratio for instances of length $n$} is
\[
\gamma_{n}= \sup_{\pi} \inf_{F\in \mathcal{F}} \frac{v_n(\pi,F)}{E_n(F)},
\]
where $E_n(F)=\sum_{\ell=1}^n\EE\left[\max\{X_1,\ldots,X_{\ell}\}\right]=\int_0^{\infty}(n-\sum_{\ell=1}^n F(x)^{\ell})\mathrm dx$ is the optimal offline value. For convenience we define $E_0(F)=0$. From stochastic dynamic programming~\citep{puterman2014markov}, it is known that we can exchange the supremum with the infimum in $\gamma_n$. Furthermore, given $F$ and $n$, the optimal value $G_n(F)=\sup_{\pi}v_n(\pi,F)$ holds the optimality equation
\[
G_n(F)= \EE\left[ \max_{1\leq \tau\leq n}\{ \tau \cdot X_n + G_{n-\tau}(F)  \}  \right]
\]
where $G_0(F)=0$. Then, the optimal policy is the one that upon observing $X_n$, chooses $\tau$ maximizing the term inside the expectation in $G_n(F)$.~\citet[Theorem 1]{abels2023prophet} show that the sequence $(G_n(F))_{n\in \NN}$ is given by the following recurrence: $G_0(F)=0$, $G_1(F)=\EE[X]$, and $$G_{n+1}(F)=\EE[X]+\EE[\max(G_n(F),nX)],$$ where $X$ is distributed according to $F$. From here, the optimal algorithm has a threshold structure of the form: At time $n\geq 1$, accept $X_n$ for $\tau=1$ units of time if $X_n < G_{n-1}(F)/(n-1)$ and for $\tau=n$ units of time if $X_n \geq G_{n-1}(F)/(n-1)$. We note that $\gamma_n=\inf_{F\in \calF}G_n(F)/E_n(F)$ and the \emph{worst-case competitive ratio corresponds} to $\gamma=\inf_{n\in \NN}\gamma_n$.

Given the threshold structure of the optimal policy, we also analyze the competitive ratio of algorithms that use a \emph{limited number of thresholds}. A \emph{$k$-threshold algorithm} partitions the time horizon $[1,n]$ into intervals $I_k, \ldots, I_1$. It assigns a value $z_i$ to each interval $I_i$, and upon observing $X_t$ with $t \in I_i$, it accepts $X_t$ for $t$ units of time if $X_t \geq z_i$; otherwise, it accepts it for only one unit of time. For a fixed $k$, we define
\[
\gamma_{n,k} = \sup_{\substack{\pi \text{ $k$-threshold} \\ {\text{algorithm} }}} \inf_{F \in \mathcal{F}} \frac{v_n(\pi, F)}{E_n(F)}.
\]

{  The following technical proposition guarantees that if we have a good competitive ratio for instances where the probability distribution $F$ is strictly increasing and infinitely differentiable, the same guarantee holds for general instances. Making $F$ continuous (and in fact infinitely differentiable) is already a standard technique using mollifiers (see, e.g.,~\cite{Liu2021}); however, the requirement that $F$ is strictly increasing is new as we need to guarantee that the derivative of $F^{-1}$ exists and it is strictly positive. We defer the proof of the following proposition to Appendix~\ref{app:small_thresholds}.

\begin{proposition}\label{prop:smooth_instances}
    Let $\pi$ be a policy that guarantees a competitive ratio of $\beta>0$ in the {\normalfont{POT}} problem for all distributions $F$ that are strictly increasing and infinitely differentiable. Then, $\gamma \geq \beta$.
\end{proposition}

A consequence of the proposition is that we can focus on studying competitive ratios on distributions $F$ that are strictly increasing and infinitely differentiable. From here, we can assume that $F^{-1}$ exists and is infinitely differentiable, as guaranteed by standard inverse function theorems. }

}

\section{Improved Guarantees for Small Number of Thresholds}\label{sec:small_thresholds}

{In this section, we provide improved analyses for algorithms that use a small number of thresholds. Similar to~\cite{perez2022iid}, we derive lower bounds via quantile-based algorithms. Our contributions are twofold: we improve upon the bounds in~\cite{abels2023prophet}, and we quantify the benefit of using a limited number of thresholds. Our results show that even a small number of thresholds yields remarkably strong performance. We begin by establishing a general, instance-independent lower bound on $\gamma_{n,k}$ for any $n$ and $k$ (Lemma~\ref{lem:key_lower_bound_fixed_thresholds}). This bound is parameterized by quantiles in $[0,1]$ and provides a way to computing competitive ratios for $k$-threshold algorithms by optimizing only the interval lengths and their associated quantiles. We then use this framework to obtain  guarantees for $k=1$, complemented by a matching hard instance that shows the tightness of our analysis. We also derive guarantees for $k \leq 6$.

Throughout this section, we assume that all distributions are strictly increasing and infinitely differentiable thanks to Proposition~\ref{prop:smooth_instances}.}

\subsection{General Multiple Threshold Analysis}

{We begin by outlining the main idea behind our lower bound on $\gamma_{n,k}$, the competitive ratio of $k$-threshold algorithms. We focus on a subclass of such algorithms parameterized by quantiles and interval lengths, which is particularly amenable to analysis (see below for a formal introduction). The key property of this class is that its expected value admits a closed-form lower bound represented in terms of the derivative of the inverse CDF, namely $\int_0^1 a(v)\, r(v)\, \mathrm{d}v$, where $a(v)\geq 0$ depends only on the chosen quantiles and interval lengths, and $r(v) = (F^{-1})'(v)$ (which exists under the smoothness assumptions of Proposition~\ref{prop:smooth_instances}). Similarly, the offline benchmark $E_n(F)$ can be written as $\int_0^1 b(v)\, r(v)\, \mathrm{d}v$, where $b(v)\geq 0$ is independent of $F$; we provide explicit formulas for $a$ and $b$ below. This shared structure allows for a direct comparison, yielding
\[
\gamma_{n,k} \;\ge\; \inf_{v\in[0,1]} \frac{a(v)}{b(v)}.
\]
Notably, the RHS is completely independent of the instance $F$, which allows us to optimize over quantiles and length of intervals, reducing the complexity of the problem. The following lemma formalizes this argument, and the remainder of the subsection is devoted to its proof.}

\begin{lemma}[Key Lower Bound]\label{lem:key_lower_bound_fixed_thresholds}
    Let integers $n_1,\ldots ,n_k \geq 1$ such that $n_1+\cdots + n_k=n$ and values $0<q_k < \cdots < q_1 <1$. Then, the competitive ratio attained by $k$-threshold algorithms on instances of size $n$ holds
    \[
    \gamma_{n,k} \geq \inf_{v\in [0,1]}\left\{  \sum_{s=1}^k \prod_{\tau>s} (1-q_\tau)^{n_\tau} \frac{\min\{v,q_s\}}{g_n(v)} A_{N_s,n_s}(q_s)   \right\},
    \]
    where $\prod_{\tau > k} (1-q_\tau)^{n_\tau}=1$, $A_{\ell,\ell'}(q) = \sum_{t=0}^{\ell'-1} (\ell - t)(1-q)^{t}$ for $ \ell' \leq \ell \leq n$, and $N_s=\sum_{\tau\leq s}n_\tau$ and $g_n(v)= \int_0^v \sum_{t=1}^n t(1-q)^{t-1}\mathrm dq$.
\end{lemma}
To prove Lemma~\ref{lem:key_lower_bound_fixed_thresholds}, we consider a general class of quantile-based $k$-threshold algorithms for which the expected value can be explicitly expressed as a function of the derivative of $F^{-1}$, denoted by $r$. Fix $n_1,\ldots,n_k\geq 1$ integers such that $n_1+\cdots + n_k = n$ and quantiles $0< q_k < \cdots < q_1 < 1$. The algorithm divides the time interval $[1,n]$, into intervals $I_1,I_2,\ldots,I_k$ where $I_t = [\sum_{\tau<t} n_t, \sum_{\tau\leq t} n_\tau]$ and computes thresholds $z_t = F^{-1}(1-q_t)$, which is well defined by our assumptions over $F$. Now, from $t=n,\ldots,1$, if ${t}\in I_{t'}$ and $X_{t} \geq z_{t'}$, then the algorithm accepts the value $X_{t}$ for the remaining ${t}$ units of time, while if $X_t< z_{t'}$, then the algorithm accepts the value $X_t$ for 1 unit of time and go to $t-1$. Let $\mathbf{n}=(n_1,\ldots,n_k)$ and $\mathbf{q}=(q_1,\ldots,q_k)$. The expected value collected by the quantile-based algorithm described above, denoted $G_{n,k}(F,\mathbf{n},\mathbf{q})$, holds:

\begin{proposition}\label{prop:G_nk_few_thresholds}
For any $\mathbf{n}=(n_1,\ldots,n_k)\in \NN^k$, $\mathbf{q}=(q_1,\ldots,q_k)$ with $0<q_k<\cdots<q_1< 1$ and infinitely differentiable $F\in \mathcal{F}$, $G_{n,k}(\mathbf{n},\mathbf{q},F)$ is equal to
\begin{align}
\sum_{s=1}^k \prod_{\tau>s}(1-q_\tau)^{n_\tau} \left(  A_{N_s,n_s}(q_s)(1-q_s) \EE[X \mid { X\geq z_s  }]  + \sum_{\ell=1}^{n_s}(1-q_s)^\ell  \EE[X\mid { X < z_s }]  \right) \label{eq:G_nk}
\end{align}
where $A_{\ell',\ell}(q)$ and $N_s$ are defined in Lemma~\ref{lem:key_lower_bound_fixed_thresholds}.
\end{proposition}
In~\eqref{eq:G_nk}, the expression inside the parentheses represents the expected value obtained from observing values in the interval $I_s$: The first term corresponds to the case where a value at least $ z_s$ is successfully observed, whose expected value is multiplied by the expected covered time for $t\in I_s$, while the second term captures the expected value contributed by values smaller than $z_s$ which are always accepted for $1$ unit of time. The product $\prod_{\tau> s}(1-q_\tau)^{n_\tau}$ multiplying the parentheses corresponds to the probability of reaching interval $I_s$ which happens when no successful value was observed in the intervals $I_k,\ldots,I_{s+1}$. A formal proof of Proposition~\ref{prop:G_nk_few_thresholds} follows by a simple inductive argument and it is deferred to Appendix~\ref{app:small_thresholds}. 

Clearly, $\gamma_{n,k}\geq \inf_{F\in \mathcal{F}} G_{n,k}(\mathbf{n},\mathbf{q},F)/E_n(F) $. Thus, it is enough to provide the inequality for the ratio $G_{n,k}(\mathbf{n},\mathbf{q},F)/E_n(F)$ for any $F$. We are ready to provide the proof of Lemma~\ref{lem:key_lower_bound_fixed_thresholds}.

\begin{proof}[Proof of Lemma~\ref{lem:key_lower_bound_fixed_thresholds}]
    Since $F^{-1}(1-u)$ is decreasing and differentiable due to our assumptions, there exists $r(v) > 0 $ such that $F^{-1}(1-u)=\int_u^1  r(v)\, \mathrm{d}v$. Then, doing a change of variable and changing the order of integration in Proposition~\ref{prop:G_nk_few_thresholds}, we obtain the following lower bound
    \begin{align*}
        G_{n,k}(\mathbf{n},\mathbf{q},F) &\geq  \sum_{s=1}^k \int_0^1 r(v)\min\{v, q_s  \} A_{N_s,n_s}(q_s) \prod_{\tau>s}(1-q_\tau)^{n_\tau}\, \mathrm{d}v.
    \end{align*}
     Similarly, 
     \begin{align*}
         E_n(F) &= \sum_{t=1}^n \int_0^\infty x\cdot  t F(x)^{t-1}\, \mathrm{d}F(x)  \\
         & = \sum_{t=1}^n \int_0^1 F^{-1}(1-u) \cdot t(1-u)^{t-1}\, \mathrm{d}u\\
         &= \int_0^1 r(v) \int_0^v \sum_{t=1}^n t(1-u)^{t-1}\, \mathrm{d}u\,\mathrm{d}v = \int_0^1 r(v)g_n(v)\, \mathrm{d}v.
     \end{align*}
     The result now follows by taking the ratio between the lower bound for $G_{n,k}$ and $E_n$ and using the standard inequality $\int_0^1 a(v)\, \mathrm{d}v / \int_0^1 b(v)\, \mathrm{d}v \geq \inf_{v\in [0,1]} a(v)/b(v)$ for $a(v),b(v)>0$ for $v\in [0,1]$.
\end{proof}

In the following subsection, we use the key inequality to obtain tight guarantees for $k=1$ threshold and new guarantees for $k\in \{2,3\}$. 
The following are technical propositions needed in the following subsections; their proof is deferred to Appendix~\ref{app:small_thresholds}.

\begin{proposition}\label{prop:monotonicity_g_n}
    The following holds:
    \begin{enumerate}[itemsep=0pt,label=\normalfont(\roman*)]
        \item $g_n(v)$ is increasing for $v\in [0,1]$.
        \item $g_n(v)/v$ is decreasing for $v\in [0,1]$.
    \end{enumerate}
\end{proposition}

We define the following two auxiliary functions that will serve us as limits of $A_{N_s,n_s}$ and $g_n$,
\[
\Bar{A}_{\phi,\theta}(\alpha) = \frac{  e^{-\alpha \theta}( 1-(\phi -\theta)\alpha) + \alpha \phi -1 }{\alpha^2} \qquad \text{and} \qquad
\Bar{g}(\lambda) = \frac{e^{-\lambda}+\lambda-1}{\lambda} = \lambda \cdot \Bar{A}_{1,1}(\lambda).
\]

\begin{proposition}\label{prop:limit_of_functions_limited_thresholds}
    The following limits hold:
    \begin{enumerate}[itemsep=0pt,label=\normalfont(\roman*)]
        \item For $1\geq \phi \geq  \theta\geq 0$, and $\alpha\geq 0$, we have ${A_{\phi n, \theta n}(\alpha/n)}/{n^2} \to \Bar{A}_{\phi,\theta}(\alpha).$
        \item For $\lambda \geq 0$, ${g_n(\lambda/n)}/{n} \to \Bar{g}(\lambda).$
    \end{enumerate}
\end{proposition}

\subsection{Optimal Analysis for Single-Threshold Algorithms}

{In this subsection, we focus on single-threshold algorithms. Recall that $\gamma_{n,1}$ corresponds to the largest competitive ratio attained by a single-threshold algorithm on instances of size $n$. Our main result is the following.

\begin{theorem}\label{thm:complete_characterization_gamma_n1}
The worst-case competitive ratio for single-threshold algorithms is $\inf_{n\in \NN}\gamma_{n,1}=({1+e^{-2}})/2\approx 0.567$.
\end{theorem}

To prove this theorem, first, we use Lemma~\ref{lem:key_lower_bound_fixed_thresholds} to ensure the lower bound on $\gamma_{n,1}$ for any $n$ and then, we present a hard instance such that $\inf_n \gamma_{n,1}\leq (1+e^{-2})/2$.

\begin{lemma}\label{lem:lower_bound_single_threshold}
For any $n\in \NN$, the largest competitive ratio attained by single-threshold algorithms on instances of length $n$ holds $\gamma_{n,1}\geq \frac{n+1}{n}\left( \frac{(1-2/(n+1))^{n+1} +1 }{2} \right)\geq (1+e^{-2})/2$.
\end{lemma}

\begin{proof}
%We prove both inequalities separately. 
Using Lemma~\ref{lem:key_lower_bound_fixed_thresholds} with $k=1$, $n_1=n$ and $q_1=\alpha/(n+1)$, we obtain
\[
\gamma_{n,1} \geq \inf_{v\in [0,1]} \left\{ \frac{\min\{ v, \alpha/(n+1)  \} }{g_{n}(v)} A_{n,n}(\alpha/(n+1))\right\}.
\]
    We analyze separately the cases where $v\leq \alpha/(n+1)$ and $v> \alpha/(n+1)$. Note that for $v\leq \alpha/(n+1)$,
    \[
    \frac{\min\{ v, \alpha/(n+1)  \} }{g_{n}(v)} = \frac{v}{g_n(v)}
    \]
    By Proposition~\ref{prop:monotonicity_g_n}, this last function is increasing in $v$ so it attains its minimum at $v=0$. Furthermore, $\lim_{v\to 0} g_n(v)/v \to n(n+1)/2$ which can be easily computed using the definition of $g_n(v)$. Now, for $v\geq \alpha/(n+1)$,
    \[
    \frac{\min\{ v, \alpha/(n+1)  \} }{g_{n}(v)} = \frac{\alpha}{(n+1)g_n(v)}
    \]
    and this last function is decreasing by Proposition~\ref{prop:monotonicity_g_n}; hence, it attains its minimum at $v=1$ with $\lim_{v\to 1} g_n(v)= n$. Putting these two results together and using $\alpha=2$, we obtain
    \[
    \gamma_{n,1}\geq \frac{\min\{2,\alpha \}}{n(n+1)} A_{n,n}(\alpha/(n+1)) = \frac{A_{n,n}(2/(n+1))}{n(n+1)}= \left( \frac{n+1}{n}\right) \left( \frac{(1-2/(n+1))^{n+1} +1 }{2} \right),
    \]
    which proves the first inequality. For the second inequality, we can show that it holds for $n=1,2,3$ by inspection, so we can assume that $n\geq 4$. Now, the second inequality in the statement of the lemma holds if and only if $(n+1)\left(1-2/(n+1)\right)^{n+1} - {n}/{e^2}+1 \geq 0$
which we obtain simply by reordering terms. The LHS of this inequality can be lower bounded by
\begin{align}
(n+1)\left(1-\frac{2}{n+1}\right)^{n+1} - \frac{n}{e^2}+1 &\geq (n+1)e^{-2-4/(n-1)}-\frac{n}{e^2}+1 \nonumber \\
&\geq e^{-2}\left( (n+1)\left( 1-\frac{4}{n-1}\right) -n  \right)+1 \label{ineq:last_one}
\end{align}
where both inequalities follow by the inequality $1+x\leq e^x$ applied as follows: (1) $(1-2/(n+1))^{n+1}\geq (1+2/(n-1))^{-(n+1)}\geq e^{-2(n+1)/(n-1)}=e^{-2(1+2/(n-1))}$ and (2) it's a direct application of the inequality with $x=-4/(n-1)$. Then, a straightforward verification shows that~\eqref{ineq:last_one} is nonnegative for any $n\geq 4$, which concludes the proof of the lemma.
\end{proof}
We finish the proof of Theorem~\ref{thm:complete_characterization_gamma_n1} by providing a hard instance for single-threshold algorithms. Fix $\beta \in [1,n]$. We present the inverse of a distribution $f(u)=F^{-1}(1-u)$ as the analysis for quantiles is more amenable. Consider
\[
f(u) =\begin{cases}
    2n\left( \frac{e^2-3}{e^2+1}  \right)  & u \in [0, 1/n^3), \\
    \frac{1}{n}\left(\frac{4}{e^2+1}\right) & u\in [1/n^3, 1/n^3+\beta/n),\\
    0 & u\in [1/n^3+\beta/n,1).
\end{cases}
\]
Note that $f$ is nonincreasing and discontinuous. However, by using a regularizer, we can obtain a infinite differentiable approximation to $f$ where the same results that we present here hold, up to an error that can be made arbitrarily small.

%\begin{figure}[h!]
%    \centering
%    \begin{tikzpicture}[y=.4cm, x=.7cm,font=\sffamily]
 	%axis
%	\draw[->] (0,0) -- coordinate (x axis mid) (10,0);
%    	\draw[->] (0,0) -- coordinate (y axis mid) (0,10);

%        \draw (1,1pt) -- (1,-2pt) node[below] {\small$\varepsilon$};
%        \draw (2,1pt) -- (2,-2pt) node[below] {\small$\varepsilon'$};

%        \draw (8,1pt) -- (8,-2pt) node[below] {\small$M-\delta$};
%        \draw (9,1pt) -- (9,-2pt) node[below] {\small$M$};

%        \draw (1pt,3) -- (-2pt,3) node[left] {\small$1-k/n$};
%        \draw (1pt,6) -- (-2pt,6) node[left] {\small$1-(k/n)^2$};

%        \draw (1pt,9) -- (-2pt,9) node[left] {\small$1$};
      
	%labels      
	%\node[below=0.8cm] at (x axis mid) {MOPS};
	%\node[rotate=90, above=0.8cm] at (y axis mid) {Power [mW]};

	%plots
%	\draw[-] (0,0) -- (1,3) -- (2,6) -- (8,6) -- (9,9);

%\end{tikzpicture}
%    \caption{Piece-wise description of $F^{-1}(1-u)$.}
%    \label{fig:cdf_figure}
%\end{figure}

The following two propositions provide asymptotic exact value and upper bound on the prophet value and the single-threshold algorithm, respectively. We defer their proof to Appendix~\ref{app:small_thresholds}.

\begin{proposition}\label{prop:Formula_En_hard_dist_1_threshold}
    We have $E_n \to (e^2-3)/(e^2+1)+ 4(e^{-\beta}+\beta-1)/(\beta(e^2+1)) $ when $n\to \infty$.
\end{proposition}

\begin{proposition}\label{prop:UB_ALG_1_threshold}
    We have $$\lim_{n\to \infty} G_{n,1}\leq \max\left\{ \frac{a}{2}, \left(a + b\beta\right)\left( \frac{e^{-\beta}+\beta-1}{\beta^2} \right), \max_{\lambda\in [0,\beta]} \left\{ \left(\frac{e^{-\lambda}+\lambda-1}{\lambda^2}\right) (a+\lambda b) \right\}\right\},$$
    where $a=2(e^2-3)/(e^2+1)$ and $b=4/(e^2+1)$.
\end{proposition}

Using these two propositions, 
\[
\lim_{n\to \infty} \frac{G_{n,1}}{E_n} \leq \frac{ \max\left\{ a/2, \left(a + b\beta\right)\left( \frac{e^{-\beta}+\beta-1}{\beta^2} \right), \max_{\lambda\in [0,\beta]} \left\{ \left(\frac{e^{-\lambda}+\lambda-1}{\lambda^2}\right) (a+\lambda b) \right\}\right\}}{(e^2-3)/(e^2+1)+ 4(e^{-\beta}+\beta-1)/(\beta(e^2+1))}
\]
for any $\beta>0$. Note that the denominator in the right-hand side of the inequality tends to $1$ when $\beta\to \infty$. Hence,
\[
\lim_{\beta\to \infty}\lim_{n\to \infty} \frac{G_{n,1}}{E_n} \leq \max\left\{ \frac{e^2-3}{e^2+1}, \frac{4}{e^2+1}, \max_{\lambda\in [0,\infty)} \left\{ \left(\frac{e^{-\lambda}+\lambda-1}{\lambda^2}\right) \left(2\frac{e^2-3}{e^2+1}+\lambda \frac{4}{e^2+1}\right) \right\}\right\}
\]
The function $\lambda \mapsto (2(e^2-3)+ 4\lambda )(e^{-\lambda}+\lambda-1)/\lambda^2$ is increasing between $[0,2]$ and decreasing in $[2,\infty)$ (see Proposition~\ref{prop_app:temp_function_1_threshold} in Appendix~\ref{app:small_thresholds}). Hence, the maximum of such a function happens at $\lambda=2$. Therefore,
\[
\lim_{\beta\to \infty}\lim_{n\to \infty} \frac{G_{n,1}}{E_n}\leq \max\left\{  \frac{e^2-3}{e^2+1} , \frac{4}{e^2+1}, \frac{1}{2e^2}(e^2+1) \right\} = \frac{1+e^{-2}}{2}.
\]
Then, the inequality $\inf_{n\in \NN} \gamma_{n,1}\leq \lim_{\beta\to \infty} \lim_{n\to \infty} G_{n,1}/E_n$ finishes the proof of Theorem~\ref{thm:complete_characterization_gamma_n1}.}

\subsection{Analysis and Guarantees for Multiple Thresholds}

We consider $k\geq 1$ thresholds $0<\alpha_k/n < \alpha_{k-1}/n< \cdots < \alpha_1/n < 1$ and $n_s=n/k$ for all $s=1,\ldots,k$, where we assumed that $k$ divides $n$. This last assumption is to avoid notational clutter; otherwise, we would have $n_s\in \{ \lfloor n/k \rfloor, \lceil n/k \rceil  \}$ for all $s=1,\ldots,k$, but since $n_s/n\to 1/k$ when $n\to \infty$, the asymptotic behavior is unaltered when we focus on $n$ divisible by $k$. Then, similar to the previous cases, using Lemma~\ref{lem:key_lower_bound_fixed_thresholds} and breaking down the interval of optimization, we obtain the following:
\begin{lemma}\label{lemma:lowerbound_asymptotic_equidistant_k}
    For any fixed $0<\alpha_k < \alpha_{k-1}< \cdots < \alpha_1 < n$, we have
    \begin{align*}
    \lim_{n\to \infty} \frac{G_{n,k}}{E_n} &\geq  \min\left\{ 2\sum_{t=1}^k e^{-\sum_{\tau>t}\alpha_\tau/k }  \Bar{A}_{t/k,1/k}(\alpha_t),\,  \sum_{t=1}^k \alpha_t e^{-\sum_{\tau>t} \alpha_\tau/k} \Bar{A}_{t/k,1/k}(\alpha_t)     , \right. \\
    & \left.\inf_{\substack{j\in [k-1]\\\lambda\in [\alpha_{j+1},\alpha_j]}} \left\{ \frac{\lambda}{\Bar{g}(\lambda)} \sum_{t\leq j} e^{-\sum_{\tau>t}\alpha_\tau/k } \Bar{A}_{t/k,1/k}(\alpha_t) 
 + \frac{1}{\Bar{g}(\lambda)} \sum_{t > j} \alpha_t e^{-\sum_{\tau>t}\alpha_\tau/k } \Bar{A}_{t/k,1/k}(\alpha_t) \right\}    \right\}.
\end{align*}
\end{lemma}
We skip the proof as it follows the same construction as in the previous subsections. {Table~\ref{table:value_few_thresholds} summarizes the competitive ratio lower bounds for $k=1,\ldots, 6$ thresholds. Unfortunately, for larger number of thresholds ($k\geq 7)$ the formula on the RHS of Lemma~\ref{lemma:lowerbound_asymptotic_equidistant_k} becomes numerically unstable. In Appendix~\ref{app:numerical_table} we report the $\alpha$ values from the lemma.
\begin{table}[ht]
\centering
\label{tab:lemma4-lower-bounds}
\begin{tabular}{|c| c|}
\hline
$k$ & Competitive ratio lower bound  \\
1 & $(1+e^{-2})/2\approx 0.5676$ \\
2 & 0.5967 \\
3 & 0.6027 \\
4 & 0.6055 \\
5 & 0.6068 \\
6 & 0.6075\\
\hline
\end{tabular}
\caption{Numerical lower bounds on the competitive ratio from Lemma~\ref{lemma:lowerbound_asymptotic_equidistant_k}.}\label{table:value_few_thresholds}
\end{table}
}

%For $k=3$, via inspection on values $\alpha_3<\alpha_2<\alpha_1$, we get $\alpha_1\approx 62.74$, $\alpha_2\approx 5.55$, $\alpha_3\approx 0.960$ and $\lim_{n\to \infty} \gamma_{n,3} \geq 0.60265$.

\section{Tightness via Convex Optimization}\label{sec:cvx-optimal}

{In this section, we show that for each $n$, the worst-case competitive ratio $\gamma_n$ of the optimal dynamic programming policy can be obtained by analyzing a finite-dimensional, linearly constrained convex optimization program (see~\ref{form:cvx}). The construction of this convex program is technically involved, so we develop it in stages. First, we reformulate the condition $\gamma_n \geq 1/(1+\varepsilon)$ as the requirement that a certain mathematical program over functions $F \in \mathcal{F}$ has a nonnegative value (see~\ref{eq:PHI_problem}). Second, we show that the sign of this program is equivalent to that of a program defined over the set of values attained by the dynamic program in the worst case (see~\ref{form:infinite}). Finally, we establish that this latter program is equivalent to the finite-dimensional convex program~\ref{form:cvx}, summarizing this result in Theorem~\ref{thm:apx-characterization}. Notably, the program~\ref{form:cvx}, although linearly constrained, can be analyzed via unconstrained first-order conditions obtained by simply setting the gradient of the objective to $0$ as it's optimal solution lies in the interior of the feasible region (see Lemma~\ref{claim:alphas}). This last result allows us to compute $\gamma_n$ via a simple recursion (see~\ref{opt-alphas}).

In the remainder of the section, we expand on these steps in a constructive manner.

\paragraph{Reformulating the Competitive Ratio Question via Mathematical Programming.} Our approach begins by studying the difference between $(1+\varepsilon)G_n(F)-E_n(F)$ for $\varepsilon>0$ and $F\in \mathcal{F}$, where recall that $G_n(F)$ is the value of the dynamic program and $E_n(F)$ is the value of the prophet. We note that if the problem
\begin{align}
\inf\{  (1+\varepsilon)G_n(F) - E_n(F) : F\in \mathcal{F}  \} \tag*{$[\Phi]_{n,\varepsilon}$}\label{eq:PHI_problem}
\end{align}
has nonnegative value, then $\gamma_n \geq 1/(1+\varepsilon)$. From here, it is immediate that $\varepsilon_n= \inf\{ \varepsilon>0 : \text{\ref{eq:PHI_problem} has nonnegative value} \}$ satisfies $\gamma_n = 1/(1+\varepsilon_n)$. Thus, characterizing $\gamma_n$ reduces to be able to solve problem~\ref{eq:PHI_problem} for $\varepsilon\geq \varepsilon_n$. However, problem~\ref{eq:PHI_problem} is infinite-dimensional and non-convex. Our next steps address these challenges.

\paragraph{Introducing a Finite-Dimensional Convex Problem in The Space of Thresholds.}

We first show that the nonnegativity of problem~\ref{eq:PHI_problem} can be verified via another (still infinite-dimensional) problem that is parametrized by the value that the optimal dynamic program takes for different values of $n$. This last problem can be later reduced to a convex finite-dimensional problem via a truncation argument.

Let's explain the first reduction more concretely. First, note that for a distribution $F\in \mathcal{F}$ we can compute for every $j\geq 0$, $G_{j}(F)$ using the recursion in Section~\ref{sec:prelim}. Furthermore, we know that the optimal policy uses the thresholds $\tau_j(F)=G_j(F)/j$ such that if the value observed at time $j$ is at least $\tau_j(F)$, then, it is optimal to accept the value for the remainder of the time, i.e., $j$ units of time. A crucial result is that $\tau_j(F)\leq \tau_{j+1}(F)$ for any $j\geq 0$. Now, note that for $E_n(F)$, using Jensen's inequality, we can justify that
\begin{align*}
E_n(F) & = \int_0^\infty \left(  n - \sum_{\ell=1}^n F(x)^\ell \right)\, \mathrm{d}x \tag{See Section~\ref{sec:prelim}}\\
& = \sum_{j=0}^\infty \int_{\tau_j(F)}^{\tau_{j+1}(F)} \left( n- \sum_{\ell=1}^n F(x)^\ell  \right) \, \mathrm{d}x \\
& \leq \sum_{j=0}^\infty (\tau_{j+1}(F) - \tau_j(F)) P_n\left(\frac{\int_{\tau_{j}(F)}^{\tau_{j+1}(F)} F(x)\, \mathrm{d}x} { \tau_{j+1}- \tau_j(F)}  \right)=\text{UB}
\end{align*}
where $P_n(t)= n - \sum_{\ell=1}^n t^\ell$. If we define the following distribution:
\[
H(x) = \begin{cases} 
\frac{1}{\tau_{j+1}(F)-\tau_{j}(F)}\int_{\tau_{j}(F)}^{\tau_{j+1}(F)} F(x)\, \mathrm{d}x & x\in [\tau_{j}(F), \tau_{j+1}(F))\\
0 & \text{otherwise
}
\end{cases}
\]
we can show that $G_j(H) = j \cdot\tau_j(F) = G_j(F)$ and that $E_n(F) \leq \text{UB} =  E_n(H)$ (see the proof of Lemma~\ref{lem:inf-opt} for a details). This result says that hard distributions for the POT problem are found in step-function CDFs. This parallels the conclusion obtained using the \emph{balayage} technique by~\cite{hill1982comparisons} in the IID prophet inequality problem.

Pushing this idea further, one may ask \emph{which properties of a sequence $\mathbf{T}=(T_j)_{j\in \NN}$ allows us to define a valid step-function CDF $H$ such that $\tau_j(H)=T_j$?} To this end, we introduce some notation. For a sequence $\mathbf{T} = (T_j)_{j \in \mathbb{N}}$, define $(\mu_j(\mathbf{T}))_{j \in \mathbb{N}}$ by
\[
\mu_0(\mathbf{T}) = 2(T_2 - T_1),
\]
and, for $j \ge 1$,
\[
\mu_j(\mathbf{T}) = \frac{j+2}{j+1} T_{j+2} - \frac{j+1}{j} T_{j+1} + \frac{T_1}{j(j+1)}.
\]
If $T_j = \tau_j(F)$, then we can prove $\mu_j(\mathbf{T}) = \int_{\tau_j(F)}^{\tau_{j+1}(F)} F(x)\, \mathrm{d}x$ (see Proposition~\ref{prop:recurrence}). Furthermore, any sequence satisfying $T_j = \tau_j(F)$ must obey the following properties (which is formally proved in the lemma below):
\begin{enumerate}[itemsep=0pt,label=\normalfont{(\Roman*)}]
    \item $T_0 = 0$, the sequence $\mathbf{T}$ is strictly increasing, $\lim_{j \to \infty} T_j / j = 0$, and $\lim_{j \to \infty} (T_{j+1} - T_j) = 0$; \label{mon1}
    \item $\mu_j(\mathbf{T}) \le T_{j+1} - T_j$ for every $j$; \label{mon2}
    \item the ratio $\mu_j(\mathbf{T}) / (T_{j+1} - T_j)$ is non-decreasing in $j$. \label{mon3}
\end{enumerate}
Thus, it is natural to enforce these conditions on the sequence $\mathbf{T}$. This leads to the following key lemma, which reduces the verification of the nonnegativity of the optimal value of~\ref{eq:PHI_problem} over all distributions to a problem over sequences that satisfy~\ref{mon1}-\ref{mon3}.
\begin{lemma}\label{lem:inf-opt}
The optimal value of problem~\ref{eq:PHI_problem} is nonnegative if and only if the optimal value of the problem
\begin{equation}
\min\left\{(1+\varepsilon)nT_n-\sum_{j=0}^{\infty}(T_{j+1}-T_{j})P_n\hspace{-.1cm}\left(\frac{\mu_j(\mathbf{T})}{T_{j+1}-T_j}\right)\hspace{-.05cm}: \mathbf{T}=(T_j)_{j\in \NN}\text{ satisfies \ref{mon1}-\ref{mon3} and }T_1=1\right\} \tag*{\normalfont{\mbox{[$\calI$]$_{n,\varepsilon}$}}}\label{form:infinite}
\end{equation}
is nonnegative.
\end{lemma}
{ We defer the formal proof of the lemma to Subsection~\ref{subsec:inf-opt}. In essence, one direction has already been established by the analysis above (up to minor technical details provided in the proof), while the other direction follows by showing that conditions~\ref{mon1}--\ref{mon3} define a valid instance of problem~\ref{eq:PHI_problem}: if $T_j = \tau_j(F)$ for a step-function CDF $F$ defined via $\mathbf{T}$, then $G_n(F) = n T_n$ and $E_n(F) = \sum_{j=0}^\infty (T_{j+1} - T_j), P_n\left(\frac{\mu_j(\mathbf{T})}{T_{j+1} - T_j}\right)$, so that the objective function in~\ref{form:infinite} is precisely $(1+\varepsilon)G_n(F) - E_n(F)$. The normalization $T_1 = 1$ is arbitrary and can be replaced by any positive value.}

Armed with problem~\ref{form:infinite}, we now focus on reducing it to a convex finite-dimensional problem. The first step corresponds to showing that we can truncate sequences $\mathbf{T}=(T_j)_{j\in \NN}$ to only the values in $T_0,T_1,\ldots,T_n$.

The following result provides a way to bound the tail contribution to the objective in the optimization problem~\ref{form:infinite}, that only depends on the values $T_1,\ldots,T_n$. 
%We defer its proof to Section \ref{subsec:techlemmas}.
%This allows, in particular, to reduce the number of variables for the problem.
\begin{lemma}\label{lem:truncating-sequence}
For every sequence $T$ with $T_1=1$ that satisfies \ref{mon1}-\ref{mon3}, and such that $\mu_{n-1}(T)/(T_{n}-T_{n-1})<1$, there exists $\delta<1$ such that 
$$\sum_{j=n-1}^{\infty}(T_{j+1}-T_{j})P_n\left(\frac{\mu_j(\mathbf{T})}{T_{j+1}-T_j}\right)\le \frac{P_n(\delta)}{1-\delta}\left(\frac{n}{n-1}T_n-T_{n-1}-\frac{1}{n-1}\right).$$
%$$\sum_{j=n-1}^{\infty}(T_{j+1}-T_{j})P_n\left(\frac{\mu_j(\mathbf{T})}{T_{j+1}-T_j}\right)\le \frac{n(n+1)}{2}\left(\frac{n}{n-1}T_n-T_{n-1}-\frac{1}{n-1}\right).$$
\end{lemma}
The proof is quite technical and follows by using properties~\ref{mon1}--\ref{mon3}. We defer the proof of this lemma to Section~\ref{subsec:techlemmas}. 

A routine calculation shows that $P_n(\delta)/(1-\delta)\leq n(n+1)/2$ since the LHS of this inequality is increasing in $\delta$. Thus, Lemma~\ref{lem:truncating-sequence} allows us to formulate the following lower bound relaxation of~\ref{form:infinite}:
\[
\min_{\substack{\mathbf{T}=(T_0,\ldots,T_n)\\ T_1=1, 0= T_0\leq \cdots \leq T_n\\
\mathbf{T} \text{holds~\ref{mon2}-\ref{mon3} } }}\left\{ (1+\varepsilon)n T_n - \sum_{j=0}^{n-2}(T_{j+1}-T_j)P_n\left(\frac{\mu_j(\mathbf{T})}{T_{j+1}-T_j} \right) - \frac{n(n+1)}{2}\left(\frac{n}{n-1} T_n - T_{n-1}- \frac{1}{n-1} \right) \right\}
\]
This relaxation is linearly constrained; however, it's objective is not convex. We can reformulate this relaxation into a convex problem by performing the change of variable $y_j=T_{j+1}-T_j$. We introduce some notation to keep the exposition brief. For $j\in \{1.\ldots,n-2\}$, let $A_j:\RR^{n-2}\to \RR$ be the function defined as
\begin{align*}
A_{j}(\mathbf{y})&=\frac{j+2}{j+1}y_{j+1}-\frac{1}{j(j+1)}\sum_{\ell=1}^jy_{\ell}.
\end{align*}
Let $K_n$ be the polyhedron in $\RR^{n-1}$ defined as follows:
$$K_n=\Big\{\mathbf{y}\in \RR_+^{n-1}:A_j(\mathbf{y})\ge 0\text{ for all }j\in \{1,\ldots,n-2\}\Big\}.$$
Then, a quick calculation shows that the relaxation of~\ref{form:infinite} is equivalent to
\begin{equation}\min\Big\{\Upsilon_{n,\varepsilon}(\mathbf{y}):\mathbf{y}\in K_n\Big\}\tag*{\normalfont{\mbox{[$\calC$]$_{n,\varepsilon}$}}}\label{form:cvx}\end{equation}
where 
\[
\Upsilon_{n,\varepsilon}(\mathbf{y}) =(1+\varepsilon)n\left(1+\sum_{j=1}^{n-1}y_j\right)-\frac{n(n+1)}{2(n-1)}\left(n\sum_{j=1}^{n-1}y_j-(n-1)\sum_{j=1}^{n-2}y_j\right) -P_{n}(2y_1)-\sum_{j=1}^{n-2}y_jP_{n}\left(\frac{A_j(\mathbf{y})}{y_j}\right)
\]
The important remark is that the function $\Upsilon_{n,\varepsilon}$ is convex in the interior of $K_n$ (see Proposition \ref{prop:unique-solution} in Appendix~\ref{app:cvx-optimal}) and so~\ref{form:cvx} is a finite-dimensional convex optimization problem.

Note that if~\ref{form:cvx} has nonnegative value, then~\ref{form:infinite} has nonnegative value and so $\varepsilon\geq \varepsilon_n$. Before we prove that~\ref{form:cvx} and~\ref{form:infinite} have the same objective value, we characterize the optimal solution of~\ref{form:cvx} in a right-neighborhood of $0$. Later we show that $\varepsilon_n$ belongs to this neighborhood, implying that~\ref{form:cvx} and~\ref{form:infinite} have the same value for $\varepsilon=\varepsilon_n$.

\paragraph{Solving~\ref{form:cvx} via FO Conditions.} 

To solve~\ref{form:cvx} we could apply Lagrangian multipliers; however, this would require the introduction of $n-2$ multipliers per constraint in $K_n$ posing a non-trivial problem. Instead, we show that there is $\varepsilon_n'>0$  such that in that for $\varepsilon \in (0,\varepsilon_n']$ problem $\min\{\Upsilon_{n,\varepsilon} (\mathbf{y}): \mathbf{y}\geq 0 \}$ has an optimal solution in the interior of $K_n$; thus, also becoming an optimal solution to problem~\ref{form:cvx}.

We formalize the last paragraph by studying the first-order conditions of $\min\{ \Upsilon_{n,\varepsilon}(\mathbf{y}): \mathbf{y}\geq0 \}$. Before we provide our main result, we introduce an auxiliary system that is obtained by applying the first-order to the unconstrained optimization problem. Given $\varepsilon\ge 0$, consider the following system in the variables $\alpha=(\alpha_0,\alpha_1,\ldots,\alpha_{n-2})$: 
\begin{align}
P_n'(\alpha_{n-2})&=(n-1)(1+\varepsilon)-n(n+1)/2,\tag*{\normalfont{\mbox{[S]$_{n,\varepsilon}$}}}\label{opt-alphas}\\
P_n'(\alpha_{n-3})-P_n'(\alpha_{n-2})&=\left(\frac{n(n+1)}{2}-\beta_n(\alpha_{n-2})\right)\frac{n-2}{n-1},\notag\\
P_n'(\alpha_{j})-P_n'(\alpha_{j+1})&=\Big(\beta_n(\alpha_{j+2})-\beta_n(\alpha_{j+1})\Big)\frac{j+1}{j+2},\text{ for }j\in \{0,\ldots,n-4\},\notag
\end{align}
with $\beta_n(t)=P_n(t)-tP_n'(t)$ and $P_n(t)=n - \sum_{\ell=1}^n t^\ell$. In the following lemma, we provide a structural result linking the previous system to the optimality conditions of the convex program \ref{form:cvx}.
\begin{lemma}\label{claim:alphas} There exists $\varepsilon'_n >0 $ such that for every $\varepsilon <\varepsilon_n'$, the following holds:
\begin{enumerate}[itemsep=0pt,label=\normalfont{(\roman*)}]
    \item There is a unique non-negative vector $\alpha(\varepsilon,n)=(\alpha_0(\varepsilon,n),\ldots,\alpha_{n-2}(\varepsilon,n))$ that satisfies \ref{opt-alphas}.\label{alpha-i}
    \item $0<\alpha_0(\varepsilon,n)<\alpha_1(\varepsilon,n)<\cdots<\alpha_{n-2}(\varepsilon,n)<1$.\label{alpha-ii}

    \item {  Given ${\alpha}(\varepsilon,n)$, define
$\mathbf{y}(\varepsilon,n)
=(y_1(\varepsilon,n),\ldots,y_{n-1}(\varepsilon,n))$
recursively by $y_1(\varepsilon,n) ={\alpha_0(\varepsilon,n)}/{2}$,
and
\[
y_{j+1}(\varepsilon,n)
=
\frac{j+1}{j+2}
\alpha_j(\varepsilon,n)y_j(\varepsilon,n)
+
\frac{1}{j(j+2)}
\sum_{k=1}^{j}y_k(\varepsilon,n),
\qquad
j\in\{1,\ldots,n-2\}.
\]
Then $\mathbf{y}(\varepsilon,n)$ belongs to the interior of
$K_n$ and satisfies $\nabla\Upsilon_{n,\varepsilon}
\bigl(\mathbf{y}(\varepsilon,n)\bigr)=0$. In particular, $\mathbf{y}(\varepsilon,n)$ is the optimal solution
of problem $[C]_{n,\varepsilon}$.\label{alpha-iii} }
    
%    \item Let $y_{\ell}(\varepsilon,n)=x_{\ell+1}(\varepsilon,n)/(\ell+1)-x_{\ell}(\varepsilon,n)/\ell$ for each $\ell\in \{1,\ldots,n-1\}$, where 
%    \begin{align}
%    x_1(\varepsilon,n)&=1,\tag*{\normalfont{\mbox{[FO]$_{n,\varepsilon}$}}}\label{opt-conditions}\\
%    x_2(\varepsilon,n)&=2+\alpha_0(\varepsilon,n),\text{ and }\notag\\
%    \frac{x_j(\varepsilon,n)}{j-1}&=\alpha_{j-2}(\varepsilon,n)\left(\frac{x_{j-1}(\varepsilon,n)}{j-1}-\frac{x_{j-2}(\varepsilon,n)}{j-2}\right)+\frac{x_{j-1}(\varepsilon,n)}{j-2}-\frac{x_1(\varepsilon,n)}{(j-1)(j-2)}\notag
%\end{align}
%for every $j\in \{3,\ldots,n\}$. Then, $\mathbf{y}(\varepsilon,n)=(y_1 (\varepsilon,n), \ldots, y_{n-1} (\varepsilon,n) )$ is in the interior of $K_n$ and $\nabla \Upsilon_{n,\varepsilon}(\mathbf{y}(\varepsilon,n))=0$. In particular, $\mathbf{y}(\varepsilon, n)$ is the optimal solution of~\ref{form:cvx}.\label{alpha-iii}
\item The optimal value of~\ref{form:cvx} is strictly positive for $\varepsilon=\varepsilon_n'$.\label{alpha-iv}
\end{enumerate}
\end{lemma}
{ The System~\ref{opt-alphas} is obtained by computing the first-order conditions, i.e., $\nabla \Upsilon_{n,\varepsilon}(\mathbf{y}(\varepsilon,n))=0$ and performing the transformation in~\ref{alpha-iii}}. The proof of Lemma~\ref{claim:alphas} is technical and requires understanding the solutions $\alpha$'s to~\ref{opt-alphas} as a function of $\varepsilon$. We defer the proof to Section~\ref{subsec:techlemmas}. However, we will use the lemma to close the loop and guarantee that~\ref{form:infinite} and~\ref{form:cvx} have the same value for $\varepsilon\leq \varepsilon_n'$.

\paragraph{Closing the Gap and Characterizing $\varepsilon_n$ via~\ref{form:cvx}} Our main result is the following reformulation and characterization of $\varepsilon_n$.

\begin{theorem}\label{thm:apx-characterization}
For any $0<\varepsilon \leq \varepsilon_n'$, problem~\ref{form:infinite} and problem~\ref{form:cvx} have the same value. Furthermore, we have $\varepsilon_n < \varepsilon_n'$.
\end{theorem}

{A byproduct of the theorem is that~\ref{form:infinite} is nonnegative if and only if~\ref{form:cvx} is nonnegative. Equivalently $\gamma_n\geq 1/(1+\varepsilon)$ if and only if~\ref{form:cvx} has nonnegative value. However, for $\varepsilon > \varepsilon_n'$, we do not have a clean characterization of the optimal solution to~\ref{form:cvx}, since we cannot guarantee that $\mathbf{y}(\varepsilon,n)$ belongs to the interior of $K_n$, and to describe the optimal solution to~\ref{form:cvx}, we would be forced to study the Lagrange multipliers for constraints of the form $A_j(\mathbf{y}) \geq 0$. As a result, we cannot guarantee that problems~\ref{form:infinite} and~\ref{form:cvx} have the same objective value for $\varepsilon>\varepsilon_n'$. In any case, that interval $(\varepsilon_n',+\infty)$ is of less importance as we are only interested in $\varepsilon = \varepsilon_n < \varepsilon_n'$, we know that~\ref{form:cvx} is nonnegative for $\varepsilon=\varepsilon_n$ and it can be found using bisection.

{ The theorem guarantees that we can solve~\ref{opt-alphas} for $\varepsilon=\varepsilon_n$ and obtain $\gamma_n = 1/(1+\varepsilon_n)$ by seeking $\alpha$ that satisfies~\ref{opt-alphas} and $\Upsilon_{\varepsilon_n,n}(\mathbf{y})=0$. This last requirement is equivalent to $\beta_n(\alpha_0(\varepsilon_n,n))=(1+\varepsilon_n)n$.} In Subsection~\ref{app:asymptotic-cvx}, we provide a limit analysis in $n$ to derive guarantees that are independent of $n$ given by System~\eqref{ode:1}-\eqref{ode:3}.}

{ 
We now turn to the proof of Theorem~\ref{thm:apx-characterization}.
Fix $\varepsilon\in(0,\varepsilon_n']$. It remains to construct,
from the optimal solution $\mathbf{y}$ of problem~\ref{form:cvx}, a
family of sequences satisfying~\ref{mon1}--\ref{mon3} whose objective
values in problem~\ref{form:infinite} converge to
$\Upsilon_{n,\varepsilon}(\mathbf{y}(\varepsilon,n))$.

Recall that problem~\ref{form:cvx} was obtained using the change of
variables $y_j=T_{j+1}-T_j$. Setting $x_j=jT_j$, this relation becomes $y_j={x_{j+1}}/{(j+1)}-{x_j}/{j}$.
We now reverse this transformation. Let $\alpha(\varepsilon,n)$ be the solution to~\ref{opt-alphas}. Using
$\alpha_0(\varepsilon,n)=2y_1(\varepsilon,n)$ and $\alpha_j(\varepsilon,n)=A_j(\mathbf{y}(\varepsilon,n))/y_j(\varepsilon,n)$ for $j\geq1$,
define $x_1(\varepsilon,n),\ldots,x_n(\varepsilon,n)$ recursively by
\begin{align}
    x_1(\varepsilon,n)&=1,\tag*{\normalfont{\mbox{[FO]$_{n,\varepsilon}$}}}\label{opt-conditions}\\
    x_2(\varepsilon,n)&=2+\alpha_0(\varepsilon,n),\text{ and }\notag\\
    \frac{x_j(\varepsilon,n)}{j-1}&=\alpha_{j-2}(\varepsilon,n)\left(\frac{x_{j-1}(\varepsilon,n)}{j-1}-\frac{x_{j-2}(\varepsilon,n)}{j-2}\right)+\frac{x_{j-1}(\varepsilon,n)}{j-2}-\frac{x_1(\varepsilon,n)}{(j-1)(j-2)}\notag
\end{align}
For every $\eta \in (0,1)$, consider the sequence $\mathbf{X}_{\eta}(\varepsilon,n) = (X_{j,\eta}(\varepsilon,n))_{j \in \NN}$ defined as follows: $X_{0,\eta}(\varepsilon,n) = 0$, $X_{j,\eta}(\varepsilon,n) = x_j(\varepsilon,n)/j$ for every $j \in \{1,\ldots,n\}$, and
\begin{equation*}
X_{j+2,\eta}(\varepsilon,n) = \frac{j+1}{j+2}\left(\eta\big(X_{j+1,\eta}(\varepsilon,n) - X_{j,\eta}(\varepsilon,n)\big) + \frac{j+1}{j}X_{j+1,\eta}(\varepsilon,n) - \frac{x_1(\varepsilon,n)}{j(j+1)}\right),
\end{equation*}
for every integer $j \ge n-1$.
}

The following lemma shows that the sequence $(X_{j,\eta}(\varepsilon,n))_{j \in \NN}$ is feasible for~\ref{form:infinite}.

\begin{lemma}\label{lem:extension-to-infinite}
For every $\varepsilon \in (0,\varepsilon_n']$, there exists $\eta_0(\varepsilon,n) \in (0,1)$ such that for every $\eta \in [\eta_0(\varepsilon,n),1)$, the sequence $(X_{j,\eta}(\varepsilon,n))_{j \in \NN}$ satisfies conditions~\ref{mon1}--\ref{mon3}.
\end{lemma}

The proof of this lemma is a simple verification, deferred to Subsection~\ref{subsec:techlemmas}. We are ready to prove Theorem~\ref{thm:apx-characterization}.

\begin{proof}[Proof of Theorem \ref{thm:apx-characterization}]

First, we just need to prove that the value of \ref{form:infinite} is upper-bounded by the optimal value of \ref{form:cvx}.
For every value $\eta\in (0,1)$ consider the sequence $(X_{j,\eta}(\varepsilon,n))_{j\in \NN}$ as defined above.
By Lemma \ref{lem:extension-to-infinite}, the sequence $(X_{j,\eta}(\varepsilon,n))_{j\in \NN}$ satisfies the properties \ref{mon1}-\ref{mon3} for every $\eta\in [\eta_0(\varepsilon,n),1)$, and therefore, it is feasible for the optimization problem \ref{form:infinite}.
Observe that by construction, for each $\eta$, we have $\mu_j(X_{\eta}(\varepsilon,n))/(X_{j+1,\eta}(\varepsilon,n)-X_{j,\eta}(\varepsilon,n))=\eta$ for every $j\ge n-1$. 
\begin{claim}\label{claim:tailsum}
For each $\eta\in (0,1)$, we have 
\begin{align*}
\sum_{j=n-1}^{\infty}(X_{j+1,\eta}(\varepsilon,n)-X_{j,\eta}(\varepsilon,n))P_n(\eta)&=\frac{P_n(\eta)}{1-\eta}\cdot \frac{x_n(\varepsilon,n)-x_{n-1}(\varepsilon,n)-x_1(\varepsilon,n)}{n-1},
\end{align*}
\end{claim}
We defer the proof of Claim \ref{claim:tailsum} to Appendix \ref{app:cvx-optimal}.
Therefore, the objective value of $(X_{j,\eta}(\varepsilon,n))_{j\in \NN}$ in \ref{form:infinite} satisfies the following:
\begin{align}
&(1+\varepsilon)nX_{n,\eta}(\varepsilon,n)-\sum_{j=0}^{\infty}(X_{j+1,\eta}(\varepsilon,n)-X_{j,\eta}(\varepsilon,n))P_n\left(\frac{\mu_j(X_{\eta}(\varepsilon,n))}{X_{j+1,\eta}(\varepsilon,n)-X_{j,\eta}(\varepsilon,n)}\right)\notag\\
&=(1+\varepsilon)x_{n}(\varepsilon,n)-\sum_{j=0}^{n-2}(X_{j+1,\eta}(\varepsilon,n)-X_{j,\eta}(\varepsilon,n))P_n\left(\frac{\mu_j(X_{\eta}(\varepsilon,n))}{X_{j+1,\eta}(\varepsilon,n)-X_{j,\eta}(\varepsilon,n)}\right)\notag\\
&\hspace{5cm} - \frac{P_n(\eta)}{1-\eta}\cdot \frac{x_n(\varepsilon,n)-x_{n-1}(\varepsilon,n)-x_1(\varepsilon,n)}{n-1},\label{eq:eta-to-cvx}
\end{align}
where the equality holds since the summation from $n-1$ is given by Claim \ref{claim:tailsum}, and $X_{n,\eta}(\varepsilon,n)=x_n(\varepsilon,n)/n$.
Now we study the summation term between zero and $n-2$.
%\vvcom{Lo que esta rojo sera actualizado a la notacion de la function objetivo actual}
For $j=0$, we have
\begin{align*}
(X_{1,\eta}(\varepsilon,n)-X_{0,\eta}(\varepsilon,n))P_n\left(\frac{\mu_0(X_{\eta}(\varepsilon,n))}{X_{1,\eta}(\varepsilon,n)-X_{0,\eta}(\varepsilon,n)}\right)&=x_{1}(\varepsilon,n)P_n\left(\frac{x_{2}(\varepsilon,n)-2x_1(\varepsilon,n)}{x_{1}(\varepsilon,n)}\right)\\
&=P_n\left(x_{2}(\varepsilon,n)-2x_1(\varepsilon,n)\right)=P_n(2y_1(\varepsilon,n)),
\end{align*}
where the first equality holds since $x_1(\varepsilon,n)$. For each $j\in \{1,\ldots,n-2\}$, we have $X_{j+1,\eta}(\varepsilon,n)-X_{j,\eta}(\varepsilon,n)=y_{j}(\varepsilon,n)$, and 
\begin{align*}
\mu_j(X_{\eta}(\varepsilon,n))&=\frac{j+2}{j+1}X_{j+2,\eta}(\varepsilon,n)-\frac{j+1}{j}X_{j+1\eta}(\varepsilon,n)+\frac{X_{1,\eta}(\varepsilon,n)}{j(j+1)}\\
&=\frac{j+2}{j+1}y_{j+1}(\varepsilon,n)-\frac{1}{j(j+1)}\sum_{\ell=1}^jy_{\ell}(\varepsilon,n)=A_{j}(y_1(\varepsilon,n),\ldots,y_{j+1}(\varepsilon,n)).
\end{align*}
Therefore,
\begin{align*}
&(X_{j+1,\eta}(\varepsilon,n)-X_{j,\eta}(\varepsilon,n))P_n\left(\frac{\mu_j(X_{\eta}(\varepsilon,n))}{X_{j+1,\eta}(\varepsilon,n)-X_{j,\eta}(\varepsilon,n)}\right)\\
&=y_j(\varepsilon,n)P_n\left(A_{j}(y_1(\varepsilon,n),\ldots,y_{j+1}(\varepsilon,n))/y_j(\varepsilon,n)\right).
\end{align*}
By replacing in \eqref{eq:eta-to-cvx}, the objective value of $(X_{j,\eta}(\varepsilon,n))_{j\in \NN}$ in \ref{form:infinite} is equal to
\begin{align*}
V_{\eta}&=(1+\varepsilon)x_{n}(\varepsilon,n)-P_n(2y_1(\varepsilon,n))-\sum_{j=1}^{n-2}y_j(\varepsilon,n)P_n\left(A_{j}(y_1(\varepsilon,n),\ldots,y_{j+1}(\varepsilon,n))/y_j(\varepsilon,n)\right)\notag\\
&\hspace{5cm} - \frac{P_n(\eta)}{1-\eta}\cdot \frac{x_n(\varepsilon,n)-x_{n-1}(\varepsilon,n)-x_1(\varepsilon,n)}{n-1}\\
&=(1+\varepsilon)n\left(1+\sum_{j=1}^{n-1}y_j(\varepsilon,n)\right)-P_n(2y_1(\varepsilon,n))-\sum_{j=1}^{n-2}y_j(\varepsilon,n)P_n\left(\frac{A_{j}(y_1(\varepsilon,n),\ldots,y_{j+1}(\varepsilon,n))}{y_j(\varepsilon,n)}\right)\notag\\
&\hspace{5cm} - \frac{P_n(\eta)}{1-\eta}\cdot \frac{1}{n-1}\left(n\sum_{j=1}^{n-1}y_j(\varepsilon,n)-(n-1)\sum_{j=1}^{n-2}y_j(\varepsilon,n)\right)\\
&=\Upsilon_{n,\varepsilon}(y(\varepsilon,n))+\frac{x_n(\varepsilon,n)-x_{n-1}(\varepsilon,n)-x_1(\varepsilon,n)}{n-1}\left[\frac{n(n+1)}{2}-\frac{P_n(\eta)}{1-\eta}\right],
\end{align*}
where the second and third equalities hold by noting that $x_n(\varepsilon,n)-x_{n-1}(\varepsilon,n)-x_1(\varepsilon,n)=n\sum_{j=1}^{n-1}y_j(\varepsilon,n)-(n-1)\sum_{j=1}^{n-2}y_j(\varepsilon,n)$.
%(1+\varepsilon)n\left(1+\sum_{j=1}^{n-1}y_j\right)-\frac{n(n+1)}{2(n-1)}\left(n\sum_{j=1}^{n-1}y_j-(n-1)\sum_{j=1}^{n-2}y_j\right)
\begin{claim}\label{claim:mon-x}
$x_n(\varepsilon,n)\ge x_{n-1}(\varepsilon,n)+x_1(\varepsilon,n)$.
\end{claim}
We defer the proof of Claim \ref{claim:mon-x} to Appendix \ref{app:cvx-optimal}.
Then, using the above equality, we get that the optimal value of \ref{form:infinite} is upper bounded by
$$\inf_{\eta\in (\eta_0(\varepsilon,n),1)}V_{\eta}=\Upsilon_{n,\varepsilon}(\mathbf{y}(\varepsilon,n)),$$
where the equality holds by Claim \ref{claim:mon-x} and since $P_{n}(\eta)/(1-\eta)$ is at most $n(n+1)/2$ for $\eta\in (0,1)$.

Using~\ref{alpha-iv} in Lemma~\ref{claim:alphas}, we get that the optimal value of~\ref{form:cvx} is strictly positive when $\varepsilon=\varepsilon_n'$. Since this value is equal to the optimal value of~\ref{form:infinite} for $\varepsilon=\varepsilon_n'$, we conclude that $\varepsilon_n < \varepsilon_n'$.
\end{proof}
We remark that our analysis allows the construction of approximately worst-case distributions. 
Given a positive integer $n$ and $\eta\in (\eta_0(\varepsilon_n,n),1)$ as in Lemma \ref{lem:extension-to-infinite}, consider $(X_{j,\eta}(\varepsilon_n,n))_{j\in \NN}$ constructed from solving the system \ref{opt-conditions}.
Lemma \ref{lem:extension-to-infinite} guarantees that $(X_{j,\eta}(\varepsilon_n,n))_{j\in \NN}$ satisfies the monotonicity conditions \ref{mon1}-\ref{mon3}, and therefore, Lemma \ref{lem:sequence-to-dist} in Section \ref{subsec:inf-opt} shows how to construct a distribution $F_{\eta}$ for which the sequence of thresholds that implement the optimal policy is exactly $(X_{j,\eta}(\varepsilon_n,n))_{j\in \NN}$.
As we showed in the proof of Theorem \ref{thm:apx-characterization}, by taking $\eta\to 1$ we can approximate the value of the optimization problem \ref{form:infinite}, i.e., $(1+\varepsilon_n)G_n(F_{\eta})\approx E_n(F_{\eta})$.

}

\subsection{Limit Behavior of the System \ref{opt-alphas}}~\label{app:asymptotic-cvx}
In this section, we provide a limit analysis relating the first-order optimality conditions derived from \ref{opt-alphas} and the integro-differential equation~\eqref{ode:1}-\eqref{ode:3}. 
%We do this by analyzing the asymptotic behavior of the FO conditions of the convex problem in Section~\ref{sec:cvx-optimal} and showing that it converges to a solution to the integro-differential system. 
To keep the notation simple throughout this section, we write $P=P_n$, $\beta=\beta_n$, and $\alpha_j=\alpha_j(\varepsilon,n)$. Additionally, we avoid taking a convergent subsequence of $\varepsilon_n$ that converges to some $\varepsilon$ to keep the notation simple. That is, we assume that $\varepsilon_n \to \varepsilon$.
Now, rearranging \ref{opt-alphas} and including the condition $\Upsilon_{\varepsilon_n,n}=0$, we obtain
\begin{align}
P'(\alpha_{n-2} )&=(n-1)(1+\varepsilon_n)-n(n+1)/2, \label{eq:varphi_j_n-2} \\
\left( \frac{n-1}{n-2} \right)\left(P'(\alpha_{n-3} )-P'(\alpha_{n-2} )\right)&=\left(\frac{n(n+1)}{2}-\beta(\alpha_{n-2} )\right),\label{eq:varphi_j_n-3}\\
\left( \frac{j+2}{j+1} \right)\left(P'(\alpha_{j} )-P'(\alpha_{j+1} )\right)&=(\beta(\alpha_{j+2} )-\beta(\alpha_{j+1} )),\text{ for }j\in \{0,1,\ldots,n-4\},\label{eq:varphi_j_general}\\
\beta (\alpha_0) & = (1+\varepsilon_n)n \label{eq:varphi_j_last_condition}
\end{align}
We now sum Equations~\eqref{eq:varphi_j_n-2},~\eqref{eq:varphi_j_n-3} and Equations~\eqref{eq:varphi_j_general} for $j=i,\ldots,n-4$ and obtain:
\begin{align*}
    &P'(\alpha_{n-2}) + \left( \frac{n-1}{n-2} \right)\left(P'(\alpha_{n-3} )-P'(\alpha_{n-2} )\right) + \cdots + \left( \frac{j+2}{j+1} \right)\left(P'(\alpha_{j} )-P'(\alpha_{j+1} )\right)\\
    &= (n-1)(1+\varepsilon_n) - \frac{n(n+1)}{2} + \left(\frac{n(n+1)}{2}-\beta(\alpha_{n-2} )\right) + \sum_{i=j+1}^{n-3}(\beta(\alpha_{i+1})-\beta(\alpha_i)).
\end{align*}
Rearranging this equation and replacing $\alpha_j=1-u_j/n$ with $0\leq u_j\leq n$, we obtain
%\[
%P'(\alpha_{j}) + \sum_{i=j}^{n-3}\frac{1}{i+1}\left( P'(\alpha_i)-P'(\alpha_{i+1})  \right)= (n-1)(1+\varepsilon) - P(\alpha_{j+1}) +\alpha_{j+1}P'(\alpha_{j+1}) 
%\]
%We replace $\alpha_j = 1- u_j/n$ with $0\leq u_j\leq n$. Then, 
\begin{align}
    &\underbrace{P'\left( 1- \frac{u_j}{n} \right) - P'\left( 1- \frac{u_{j+1}}{n} \right)}_{A} + \underbrace{\sum_{i=j}^{n-3}\frac{1}{i+1}\left( P'\left( 1 - \frac{u_i}{n} \right)-P'\left( 1 - \frac{u_{i+1}}{n} \right)  \right)}_B\nonumber\\
    &= (n-1)(1+\varepsilon_n) - \underbrace{\left(P\left(1-\frac{u_{j+1}}{n} \right) + \frac{u_{j+1}}{n}P'\left( 1- \frac{u_{j+1}}{n} \right)\right)}_C \label{eq:almost_limit}
\end{align}
We normalize this equation by $1/n$ and analyze the asymptotic behavior of each term $A,B$, and $C$ separately. We let $j/n \to x\in (0,1)$ as $n\to \infty$. We also let $u(x)$ be the limit of the piece-wise function obtained by joining the points $u_1,\ldots,u_n$. Formally, for each $n$, we define the piece-wise linear function $u_n:[0,1]\to \R$ via $u_n(i/n)=u_i$ and the function is linear between $u_n((i-1)/n)$ and $u_n(i/n)$. Then $u(x)=\lim_{n\to \infty} u_n(i/n)$. Now, for $A$, we have
    {\begin{align*}
    & \quad \frac{1}{n}\left(P'\left( 1- \frac{u_j}{n} \right) - P'\left( 1- \frac{u_{j+1}}{n} \right)\right) \\ 
    & = \sum_{\ell=1}^n \frac{\ell}{n}\left( \left( 1-\frac{u_{j+1}}{n}\right)^{\ell-1} - \left( 1-\frac{u_{j}}{n}\right)^{\ell-1}  \right) \\
%    & = \sum_{\ell=1}^n \frac{\ell}{n} \frac{\left( \left( 1-\frac{u_{j+1}}{n}\right)^{\ell-1} - \left( 1-\frac{u_{j}}{n}\right)^{\ell-1}  \right)}{1/n} \frac{1}{n} \\
    & = n \left(  \frac{1-(1-\frac{u_{j+1}}{n})^n - u_{j+1}(1-\frac{u_{j+1}}{n})^n}{u_{j+1}^2} - \frac{1-(1-\frac{u_{j}}{n})^n - u_{j}(1-\frac{u_{j}}{n})^n}{u_{j}^2}     \right) \\
    & \to \left(  \frac{1- e^{-u(t)}(1+u(x))}{u(x)^2}  \right)' = h(u(x))',
    \end{align*}}
    where $h(u)=\int_0^1 t e^{-u t} \, \mathrm{d}t= (1-e^{-u}(1+u))/u^2$.

For $B$, we have
    {\begin{align*}
        &\frac{1}{n}\sum_{i=j}^{n-3}\frac{1}{i+1}\left( P'\left( 1 - \frac{u_j}{n} \right)-P'\left( 1 - \frac{u_{j+1}}{n} \right)  \right) \\
        & = \sum_{i=j}^{n-3} \frac{1}{(i+1)/n} \left(\sum_{\ell=1}^n \frac{\ell}{n} \frac{\left( \left( 1-\frac{u_{i+1}}{n}\right)^{\ell-1} - \left( 1-\frac{u_{i}}{n}\right)^{\ell-1}  \right)}{1/n} \frac{1}{n}\right) \frac{1}{n}\\
        & \to \int_x^1 \frac{1}{s} \int_0^1 t (e^{-u(s)}t)' \, \mathrm{d}t \, \mathrm{d}s = \int_x^1 \frac{1}{s} h(u(s))'\, \mathrm{d}s.
    \end{align*}}
For $C$, we have
    \begin{align*}
        \frac{1}{n}\left(P\left(1-\frac{u_{j+1}}{n} \right) +\frac{u_{j+1}}{n}P'\left( 1- \frac{u_{j+1}}{n} \right) \right) & = 1 - \sum_{\ell=1}^n \left( 1- \frac{u_{j+1}}{n}\right)^\ell \frac{1}{n} - u_{j+1} \sum_{\ell}^n \frac{\ell}{n} \left( 1- \frac{u_{j+1}}{n} \right)^{\ell-1} \frac{1}{n}\\
        & \to 1 - \int_0^1 e^{-u(x)t}\, \mathrm{d}t - u(x)h(u(x)). \end{align*} 
Then, taking limit in Equation~\eqref{eq:almost_limit}, we obtain
\begin{align}
h(u(x))' + \int_x^1 \frac{1}{s} h(u(s))'\, \mathrm{d}s - \int_0^1 e^{-u(x)t}\, \mathrm{d}t - u(x)h(u(x)) = \varepsilon, \label{ode:initial_on_u}.
\end{align}
{  Now, we prove that we must have the conditions $u(0)=+\infty$ and $u(1)=0$. For $u(1)=0$, assume by contradiction that there is $\hat{u}>0$ such that $u_{n-2}\geq \hat{u}$ for all $n$. Then, using Equation~\eqref{eq:varphi_j_n-2}, we obtain
\begin{align*}
(1+\varepsilon_n) &= \frac{1}{n-1}\left( \sum_{t=1}^n t\left( 1 - \left( 1- \frac{u_{n-2}}{n} \right)^{t-1}   \right)  \right)\\
&\geq \frac{1}{n-1} \sum_{t=\lceil n/2 \rceil+1}^n t \left( 1- e^{-\hat{u}/2}  \right) \\
& =(1-e^{-\bar{u}/2})\frac{n(n+1) - \lceil n/2\rceil (\lceil n/2 \rceil+1) }{2(n-1)}.
\end{align*}
This last expression is unbounded in $n$, which is impossible as $1+\varepsilon_n $ is finite. We conclude that $u_{n-2}\to 0$ which shows that $u(1)=0$.

For $u(0)=+\infty$, using Equation~\eqref{eq:varphi_j_last_condition}, we have
\begin{align*}
(1+\varepsilon_n) &=\frac{ \beta\left( 1-u_0/n \right)}{n} = 1 + \frac{1}{n}\sum_{t=1}^n (t-1)\left( 1- \frac{u_0}{n}  \right)^t.
\end{align*}
Suppose there is $\hat{U}>0$ such that $u_{0}\leq \hat{U}$ for all $n$. Now, for $n$ large enough, $1-u_0/n \geq e^{-2u_0/n}$, and then
\[
1+\varepsilon_n \geq 1+ \frac{1}{n}\sum_{t=1}^n (t-1)e^{-2 u_0 t/n}\geq 1+ e^{-\hat{U}} \sum_{t=\lceil n/2\rceil }^n \frac{t-1}{n} = 1+ e^{-\hat{U}}\frac{n(n-1)-(\lceil n/2\rceil-1)\lceil n/2 \rceil}{2n}.
\]
The RHS of this inequality tends to infinity when $n\to \infty$, which contradicts the fact that $1+\varepsilon_n<\infty$. Thus, $u_{0}\to \infty$ as $n\to \infty$ and so $u(0)=+\infty$.
}

Now, evaluating~\eqref{ode:initial_on_u} in $x=1$, gives us $h(u)'(1) = 1+\varepsilon$. If we derive~\eqref{ode:initial_on_u} in $x$, we obtain
\begin{align*}
    0 & = h(u(x))'' - \frac{1}{x} h(u(x))' +h(u(x)) u'(x) - u'(x) h(u(x)) - u(x)h(u(x))'\\
    & = h(u(x))'' - h(u(x))' (1/x + u(x)). 
\end{align*}
Then, by rearranging terms, we obtain
\begin{align*}
    \frac{h(u)''}{h(u)'}= \frac{1}{x}+u &\implies \ln((1+\varepsilon))- \ln (h(u)'(x)) =  - \ln x + \int_x^1 u(s) \, \mathrm{d}s \\
    & \implies h(u)'(x)= (1+\varepsilon) xe^{-\int_x^1 u(s)\, \mathrm{d}s}. 
\end{align*}
From here, the change of variable $y(x)=e^{-u(x)}$ gives us the system~\eqref{ode:1}-\eqref{ode:3}.

{ \begin{remark}
Solving numerically~\eqref{ode:1}-\eqref{ode:3} we obtain $\varepsilon\in [0.640881015104242, 0.640881015162449]$. 
\end{remark}}

\subsection{Proof of Lemma \ref{lem:inf-opt}}\label{subsec:inf-opt}
For every distribution $F\in \calF$, we denote by $\tau(F)=(\tau_{j}(F))_{j\in \NN}$ the sequence such that $\tau_0(F)=0$ and $\tau_n(F)=G_n(F)/n$ for every positive integer $n$.
We remark that the optimal dynamic programming policy can be implemented by using the sequence $(\tau_j(F))_{j\in \NN}$ as thresholds: if we have $n$ periods to go, we accept the value if it exceeds the threshold $\tau_n(F)$ \cite[Theorem 1]{abels2023prophet}.

We proceed through several steps.
First, in the following lemma, we show that for each distribution $F$ the sequence of thresholds $(\tau_j(F))_{j\in \NN}$ satisfies the monotonicity conditions \ref{mon1}-\ref{mon3} and that we can upper bound the offline benchmark $E_n(F)$ in terms of the thresholds sequence.
{Recall that for a sequence $\mathbf{T}=(T_j)_{j\in \NN}$, we define the affine transformation $(\mu_j(\mathbf{T}))_{j\in \NN}$ as follows: $\mu_0(\mathbf{T})=2(T_2-T_1)$, and
$$\mu_j(\mathbf{T})=\frac{j+2}{j+1}T_{j+2}-\frac{j+1}{j}T_{j+1}+\frac{T_1}{j(j+1)}\text{ for }j\ge 1.$$
\begin{lemma}\label{lem:dist-to-sequence}
%Every distribution satisfies the mon properties, and upper bound the max.
For every distribution $F\in \calF$, the following holds:
\begin{enumerate}[itemsep=0pt,label=\normalfont{(\roman*)}]
    \item The sequence $\tau(F)=(\tau_{j}(F))_{j\in \NN}$ satisfies conditions \ref{mon1}-\ref{mon3}.\label{f-to-mon}
    \item For every positive integer $n$, we have 
    $E_n(F)\le \sum_{j=0}^{\infty}(\tau_{j+1}(F)-\tau_{j}(F))P_n\left(\frac{\mu_j(\tau(F))}{\tau_{j+1}(F)-\tau_j(F)}\right).$\label{f-ub}
\end{enumerate}
\end{lemma}
To prove Lemma \ref{lem:dist-to-sequence}, we show first the following proposition that provides an identity satisfied by $\tau(F)$.
\begin{proposition}\label{prop:recurrence}
For every distribution $F\in \calF$, and every $j\in \NN$, we have 
$$\int_{\tau_j(F)}^{\tau_{j+1}(F)}F(x)\mathrm dx=\mu_j(\tau(F)).$$
%Furthermore, the sequence $(\tau_j(F))_{j\in \NN}$ satisfies condition \ref{mon1}, and $\lim_{j\to \infty}\tau_j(F)=\omega_1(F)$.
%Furthermore, the following holds: The sequence $(\tau_j(F))_{j\in \NN}$ is strictly increasing, $\lim_{j\to \infty}\tau_j(F)=\omega_1(F)$,
%$\lim_{j\to \infty}\tau_{j}(F)/j=0$ and $\lim_{j\to \infty}(\tau_{j+1}(F)-\tau_j(F))=0$.
\end{proposition}
\begin{proof}
We first consider $j=0$.
In this case, considering $X$ distributed according to $F$, we have
\begin{align*}
2\tau_2(F)=G_2(F)=G_1(F)+\EE[\max(G_1(F),X)]=\tau_1(F)+\EE[\max(\tau_1(F),X)].
%&=2\tau_1(F)+\EE[\max(0,X-\tau_1(F))].
\end{align*}
The chain of equalities holds from the definition of $G_2(F)$ and $G_1(F)=\tau_1(F)$.
The following claim, whose proof is deferred to Appendix \ref{app:cvx-optimal}, yields an explicit expression for the expectation of a positive part that is used several times throughout this proof.
\begin{claim}\label{claim:identity-expectation}
For every $s\in \RR_+$, we have $\EE[\max(s,X)]=\EE[X]+\int_0^sF(x)\mathrm dx$.
\end{claim}
The previous chain of equalities, $\tau_0(F)=0$, and Claim \ref{claim:identity-expectation} with $s=\tau_1(F)=\EE[X]$, imply that
$$2\tau_2(F)=2\tau_1(F)+\int_{\tau_0(F)}^{\tau_1(F)}F(x)\mathrm dx.$$
For $j\ge 1$, from the recurrence satisfied by $(G_j(F))_{j\in \NN}$, we have
\begin{align*}
(j+1)\tau_{j+1}(F)&=G_1(F)+\EE[\max(G_j(F),jX)]=\tau_1(F)+j\EE[\max(\tau_j(F),X)],
\end{align*}
and therefore, the following equality holds:
\begin{equation}
({(j+1)\tau_{j+1}(F)-\tau_1(F)})/{j}=\EE[\max(\tau_j(F),X)].\label{recur:ident}
\end{equation}
By applying \eqref{recur:ident} for both $j$ and $j+1$, we have
\begin{align*}
\mu_j(\tau(F))&=({(j+2)\tau_{j+2}(F)-\tau_1(F)})/({j+1})-({(j+1)\tau_{j+1}(F)-\tau_1(F)})/{j}\\
&=\EE[\max(\tau_{j+1}(F),X)]-\EE[\max(\tau_j(F),X)]\\
%&=\tau_{j+1}(F)+\EE[\max(0,X-\tau_{j+1}(F))]-\tau_j(F)-\EE[\max(0,X-\tau_j(F))]\\
%&=\tau_{j+1}(F)+\int_{\tau_{j+1}(F)}^{\infty}(1-F(x))\mathrm dx-\tau_{j}(F)-\int_{\tau_{j}(F)}^{\infty}(1-F(x))\mathrm dx\\
%&=\tau_{j+1}(F)+\EE[X]-\tau_{j+1}(F)+\int_{0}^{\tau_{j+1}(F)}(1-F(x))\mathrm dx-\tau_{j}(F)-\tau_1(F)+\int_{0}^{\tau_{j}(F)}(1-F(x))\mathrm dx\\
&=\EE[X]+\int_{0}^{\tau_{j+1}(F)}F(x)\mathrm dx-\EE[X]-\int_{0}^{\tau_{j}(F)}F(x)\mathrm dx=\int_{\tau_j(F)}^{\tau_{j+1}(F)}F(x)\mathrm dx,
\end{align*}
where the third equality holds from Claim \ref{claim:identity-expectation}.
\end{proof}}
%This finishes the proof of the first part of the proposition.

\begin{proof}[Proof of Lemma \ref{lem:dist-to-sequence}]
We start by proving \ref{f-to-mon}.
We first show that $\tau(F)$ is a strictly increasing sequence.
Observe that since $\omega_0(F)<\omega_1(F)$, the expectation of $F$, which is equal to $\tau_1(F)$, is strictly positive, and therefore $\tau_1(F)>0=\tau_0(F)$.
We proceed by induction; suppose that $\tau_{j}(F)<\tau_{j+1}(F)$ for every $j\in \{0,1,\ldots,\ell\}$ and $\ell\ge 0$.
Observe that from the identity shown in the first part of the proof, we have
\begin{align*}
0<\int_{\tau_{\ell}(F)}^{\tau_{\ell+1}(F)}F(x)\mathrm dx&=\mu_{\ell}(\tau(F))=\frac{\ell+2}{\ell+1}\tau_{\ell+2}(F)-\frac{\ell+1}{\ell}\tau_{\ell+1}(F)+\frac{\tau_1(F)}{\ell(\ell+1)},
\end{align*}
where the inequality holds since $\tau_{\ell}(F)$ is strictly less than $\tau_{\ell+1}(F)$.
Therefore, 
$$\tau_{\ell+2}(F)>\frac{(\ell+1)^2}{\ell(\ell+2)}\tau_{\ell+1}(F)-\frac{\tau_1(F)}{\ell(\ell+2)}=\tau_{\ell+1}(F)+\frac{\tau_{\ell+1}(F)-\tau_1(F)}{\ell(\ell+2)}>\tau_{\ell+1}(F),$$
where the last inequality follows since $\tau_1(F)<\tau_{\ell+1}(F)$.
We conclude that the sequence $(\tau_j(F))_{j\in \NN}$ satisfies condition \ref{mon1} through the following claim.
\begin{claim}\label{claim:tau-conv}
$\lim_{j\to \infty}\tau_{j}(F)/j=0$ and $\lim_{j\to \infty}(\tau_{j+1}(F)-\tau_j(F))=0$.
\end{claim}
We defer the proof of Claim \ref{claim:tau-conv} to Appendix \ref{app:cvx-optimal}.
%This finishes the proof of the proposition.
%Observe that $\lim_{j\to \infty}\tau_{j}(F)/j=0$ and $\lim_{j\to \infty}(\tau_{j+1}(F)-\tau_j(F))=0$ by Proposition \ref{prop:recurrence}.
%Condition \ref{mon1} is directly satisfied by Proposition \ref{prop:recurrence}. 
By the integral equality in Proposition \ref{prop:recurrence}, and the fact that $F(x)\le 1$ for every $x\in \RR_+$, we have
$$\frac{\mu_j(\tau(F))}{\tau_{j+1}(F)-\tau_j(F)}=\frac{1}{\tau_{j+1}(F)-\tau_j(F)}\int_{\tau_j(F)}^{\tau_{j+1}(F)}F(x)\mathrm dx\le \frac{1}{\tau_{j+1}(F)-\tau_j(F)}\int_{\tau_j(F)}^{\tau_{j+1}(F)}\mathrm dx= 1,$$
and therefore $\tau(F)$ satisfies condition \ref{mon2}.
Finally, since $F$ is non-decreasing, we have
$$F(\tau_j(F))\le \frac{1}{\tau_{j+1}(F)-\tau_j(F)}\int_{\tau_j(F)}^{\tau_{j+1}(F)}F(x)\mathrm dx\le F(\tau_{j+1}(F)),$$
and these inequalities, together with the identity from Proposition \ref{prop:recurrence} imply that $\tau(F)$ satisfies condition \ref{mon3}.

Now we proceed to prove part \ref{f-ub}.
For every $j\in \NN$, let $\Lambda_j(y)=1/(\tau_{j+1}(F)-\tau_j(F))$.
%We use the following claim.
\begin{claim}\label{claim:limitoftau}
$\lim_{j\to \infty}\tau_j(F)=\omega_1(F)$.
\end{claim}
The proof of Claim \ref{claim:limitoftau} is deferred to Appendix \ref{app:cvx-optimal}.
Now, the following chain of equalities holds:
\begin{align}
E_n(F)&=\int_0^{\infty}\Big(n-\sum_{\ell=1}^n F(x)^{\ell}\Big)\mathrm dx\notag\\
&=\sum_{j=0}^{\infty}\int_{\tau_j(F)}^{\tau_{j+1}(F)}\Big(n-\sum_{\ell=1}^nF(x)^{\ell} \Big)\mathrm dx\notag\\
&=\sum_{j=0}^{\infty}\int_{\tau_j(F)}^{\tau_{j+1}(F)}P_n(F(x))\mathrm dx=\sum_{j=0}^{\infty}(\tau_{j+1}(F)-\tau_j(F))\int_{\tau_j(F)}^{\tau_{j+1}(F)}P_n(F(x))\Lambda_j(x)\mathrm dx,\label{cvx-lem1-ineq-1}
\end{align}
where the first equality holds since $\lim_{j\to \infty}\tau_j(F)=\omega_1(F)$ by Claim \ref{claim:limitoftau} and $F(x)=1$ for every $x> \omega_1(F)$, the second equality holds by partitioning the non-negative reals according to $\tau(F)$ strictly increasing, the fourth by exchanging the integral with the finite sum, and the last equality follows from the definition of $\Lambda_{j}$ for every $j$.
{ As we showed before Lemma~\ref{lem:inf-opt}, since $P_n$ is concave and $\Lambda_j$ is a probability density function over $[\tau_j(F),\tau_{j+1}(F))$ for every $j$, Jensen's inequality implies that
\begin{align}
&\int_{\tau_j(F)}^{\tau_{j+1}(F)}P_n(F(x))\Lambda_j(x)\mathrm dx\leq P_n\left(\frac{\mu_j(\tau(F))}{\tau_{j+1}(F)-\tau_{j}(F)}\right),\label{cvx-lem1-ineq-2}
\end{align}
where the first equality holds from the definition of $\Lambda_j$ and the second by Proposition \ref{prop:recurrence}.
From inequalities \eqref{cvx-lem1-ineq-1}-\eqref{cvx-lem1-ineq-2}, we get the upper bound on $E_n(F)$, which concludes the proof.}
\end{proof}
Conversely, in the following lemma, we show that for any sequence satisfying the monotonicity properties \ref{mon1}-\ref{mon3}, there is a distribution for which this sequence gives the corresponding thresholds precisely.

\begin{lemma}\label{lem:sequence-to-dist}
%For every sequence satisfying the mon properties, there exists a distribution for which the DP coincides with the sequence, and the benchmarks are equal.
For every sequence $\mathbf{T}=(T_j)_{j\in \NN}$ that satisfies conditions \ref{mon1}-\ref{mon3}, there exists a distribution $H\in \calF$ such that $G_{\ell}(H)=\ell\cdot T_{\ell}$ for every $\ell\in \NN$, and
$$E_n(H)=\sum_{j=0}^{\infty}(T_{j+1}-T_{j})P_n\left(\frac{\mu_j(\mathbf{T})}{T_{j+1}-T_j}\right).$$
\end{lemma}

\begin{proof}
Note that condition \ref{mon2} implies that the sequence $(\mu_j(\mathbf{T})/(T_{j+1}-T_j))_{j\in \NN}$ is upper bounded by one, and non-decreasing by condition \ref{mon3}, therefore, it has a limit $\calL\in (0,1]$.
\begin{claim}\label{claim:limit-t}
When $\calL<1$, there exists a finite value $\calT$ such that $\lim_{j\to \infty}T_j=\calT$.
\end{claim}
We defer the proof of Claim \ref{claim:limit-t} to Appendix \ref{app:cvx-optimal}, and we show how to construct $H$ using this claim.
Let $H:\RR\to \RR$ be the function defined by parts as follows: for every non-negative integer $j$, let $H(x)=\mu_j(\mathbf{T})/(T_{j+1}-T_j)$ for every $x\in [T_j,T_{j+1})$, $H(x)=0$ for every $x\in (-\infty,0)$, and $H(x)=1$ for every $x\ge \lim_{j\to \infty}T_j$.
We show first that $H$ is, in fact, a distribution.
The non-negativity holds by construction, and the monotonicity condition \ref{mon3} satisfied by $T$ implies directly that $H$ is non-decreasing.
Then, it just remains to show that $\lim_{x\to \infty}H(x)=1$.
We have two cases.
When $\calL=1$ we directly get that $\lim_{x\to \infty}H(x)=\lim_{j\to \infty}\mu_j(\mathbf{T})/(T_{j+1}-T_j)=\calL=1$.
Otherwise, if $\calL<1$, by Claim \ref{claim:limit-t} we have that $(T_j)_{j\in \NN}$ converges to a finite value $\calT$ and, by construction, $H(x)=1$ for every $x\ge \calT$; thus in this case we also have that $\lim_{x\to \infty}H(x)=1$.
We conclude that $H$ is a distribution.

Now we show $G_j(H)/j=\tau_j(H)=T_j$ for every $j$.
By construction, we have $G_0(H)=T_0=0$.
\begin{claim}\label{claim:1and2} It holds $G_1(H)=T_1$.
\end{claim}
We defer the proof of Claim \ref{claim:1and2} to Appendix \ref{app:cvx-optimal}, and now proceed by induction for $j\ge 1$.
Suppose that $\tau_{j}(H)=T_{j}$ for every $j\in \{0,\ldots,\ell+1\}$ for some $\ell$; since $\tau_1(H)=G_1(H)$, Claim \ref{claim:1and2} implies that $T_1=\tau_1(H)$.
By applying Proposition \ref{prop:recurrence} with the distribution $H$, we have that 
\begin{align}
\mu_{\ell}(\tau(H))=\int_{\tau_{\ell}(H)}^{\tau_{\ell+1}(H)}H(x)\mathrm dx=\int_{T_{\ell}}^{T_{\ell+1}}H(x)\mathrm dx=\int_{T_{\ell}}^{T_{\ell+1}}\frac{\mu_{\ell}(\mathbf{T})}{T_{\ell+1}-T_{\ell}}\mathrm dx=\mu_{\ell}(\mathbf{T}),\label{lem:cvx-lem2-eq1-main}
\end{align}
where the second equality holds since $\tau_{\ell}(H)=T_{\ell}$ and $\tau_{\ell+1}(H)=T_{\ell+1}$, and the third equality follows from the definition of $H$ in $[T_{\ell},T_{\ell+1})$.
On the other hand, we have
\begin{align*}
\mu_{\ell}(\tau(H))&=\frac{\ell+2}{\ell+1}\tau_{\ell+2}(H)-\frac{\ell+1}{\ell}\tau_{\ell+1}(H)+\frac{\tau_1(H)}{\ell(\ell+1)}\\
&=\frac{\ell+2}{\ell+1}\tau_{\ell+2}(H)-\frac{\ell+1}{\ell}T_{\ell+1}+\frac{T_1}{\ell(\ell+1)}\\
&=\mu_{\ell}(\mathbf{T})+\frac{\ell+2}{\ell+1}\tau_{\ell+2}(H)-\frac{\ell+2}{\ell+1}T_{\ell+2}=\mu_{\ell}(\tau(H))+\frac{\ell+2}{\ell+1}\tau_{\ell+2}(H)-\frac{\ell+2}{\ell+1}T_{\ell+2},
\end{align*}
where the second equality holds since $\tau_{1}(H)=T_{1}$ and $\tau_{\ell+1}(H)=T_{\ell+1}$, the third by adding and subtracting a term to get $\mu_{\ell}(\mathbf{T})$, and last equality follows from \eqref{lem:cvx-lem2-eq1-main}.
Therefore, we conclude that $\tau_{\ell+2}(H)=T_{\ell+2}$, which finishes the inductive step.

Finally, we compute $E_n(H)$. 
Observe that
\begin{align}
E_n(H)&=\int_0^{\infty}P_n(H(x))\mathrm dx\notag\\
&=\sum_{j=0}^{\infty}\int_{\tau_j(H)}^{\tau_{j+1}(H)}P_n(H(x))\mathrm dx\notag\\
%&=\sum_{j=0}^{\infty}\int_{T_j}^{T_{j+1}}P_n(H(x))\mathrm dx\notag\\
&=\sum_{j=0}^{\infty}\int_{T_j}^{T_{j+1}}P_n\left(\frac{\mu_j(\mathbf{T})}{T_{j+1}-T_j}\right)\mathrm dx=\sum_{j=0}^{\infty}(T_{j+1}-T_{j})P_n\left(\frac{\mu_j(\mathbf{T})}{T_{j+1}-T_j}\right),\notag
\end{align}
where the second equality holds since $\lim_{j\to \infty}\tau_{j}(H)=\lim_{j\to \infty}T_j=\calT$ and $P_n(H(x))=0$ for every $x\ge \calT$, the third equality holds since $\tau(H)=T$, and the fourth by replacing with the definition of $H$ in every interval $[T_j,T_{j+1})$.
This finishes the proof of the lemma.
\end{proof}
By using Lemma \ref{lem:dist-to-sequence} and Lemma \ref{lem:sequence-to-dist}, we can now prove Lemma \ref{lem:inf-opt}.

\begin{proof}[Proof of Lemma~\ref{lem:inf-opt}]
Suppose the optimal value of \ref{form:infinite} is non-negative and consider any distribution $F\in \calF$.
Consider the distribution $F_{\mu}$ such that $F_{\mu}(t)=F(\mu t)$, where $\mu$ is the expectation of $F$.
Namely, $F_{\mu}$ is the normalization of $F$, so the expectation of $F_{\mu}$ is equal to one.
In particular, we have $G_n(F_{\mu})=G_{n}(F)/\mu$ and $E_n(F_{\mu})=E_n(F)/\mu$ for every $n$.
By Lemma \ref{lem:dist-to-sequence} we have 
\begin{align*}
\frac{1}{\mu}((1+\varepsilon)G_n(F)-E_n(F))&=(1+\varepsilon)G_n(F_{\mu})-E_n(F_{\mu})\\
&\ge (1+\varepsilon)n\tau_n(F_{\mu})-\sum_{j=0}^{\infty}(\tau_{j+1}(F_{\mu})-\tau_{j}(F_{\mu}))P_n\left(\frac{\mu_j(\tau(F_{\mu}))}{\tau_{j+1}(F_{\mu})-\tau_j(F_{\mu})}\right)\ge 0,
\end{align*}
where the last inequality holds since $\tau_1(F_{\mu})=\tau_1(F)/\mu=\mu/\mu=1$, and therefore $\tau(F_{\mu})$ is feasible for the optimization problem \ref{form:infinite}.

Now, for the converse, consider the feasible sequence $\mathbf{T}=(T_j)_{j\in \NN}$ feasible for \ref{form:infinite}, i.e., satisfying \ref{mon1}-\ref{mon3} and $T_1=1$.
By Lemma \ref{lem:sequence-to-dist}, there exists a distribution $H$ for which $T_j=G_j(H)/j=\tau_j(H)$, and furthermore, 
\begin{align*}
&(1+\varepsilon)nT_n-\sum_{j=0}^{\infty}(T_{j+1}-T_{j})P_n\hspace{-.1cm}\left(\frac{\mu_j(\mathbf{T})}{T_{j+1}-T_j}\right)=(1+\varepsilon)G_n(H)-E_n(H),
\end{align*}
where the equality for the summation also holds from Lemma \ref{lem:sequence-to-dist}.
Since $(1+\varepsilon)G_n(H)-E_n(H)\ge 0$, we conclude that the objective value of $T$ is non-negative in \ref{form:infinite}.
As this holds for any feasible sequence $T$, the optimal value of \ref{form:infinite} is non-negative.
\end{proof}

\subsection{Proofs of Lemma \ref{lem:truncating-sequence}, Lemma \ref{claim:alphas} and Lemma \ref{lem:extension-to-infinite}}\label{subsec:techlemmas}

{ 
\begin{proof}[Proof of Lemma~\ref{lem:truncating-sequence}]
Let $\delta_j=\mu_j(\mathbf{T})/(T_{j+1}-T_j)$. Denote $S=\sum_{j=n-1}^\infty (T_{j+1}-T_j)P_n(\delta_j)$, we aim to prove that $S \leq \frac{P_n(\delta)}{1-\delta}\left(\frac{n}{n-1}T_n-T_{n-1}-\frac{1}{n-1}\right)$ for some $\delta\in(0,1)$.
First, notice that for $k\geq n-1$, we have
\begin{align*}
\sum_{j=n-1}^k \mu_j(\mathbf{T}) & = \frac{k+2}{k+1}T_{k+2} - \frac{n}{n-1}T_n +\frac{1}{n-1}-\frac{1}{k+1}. 
\end{align*}
Then,
\begin{align*}
\sum_{j=n-1}^k (T_{j+1}-T_j)(1-\delta_j) & = \sum_{j=n-1}^k \left( T_{j+1} -T_j -\mu_j(\mathbf{T})  \right)\\
& = (T_{k+1}-T_{k+2}) +\frac{1-T_{k+2}}{k+1} +\frac{n}{n-1}T_n -T_{n-1}-\frac{1}{n-1}
\end{align*}
Using~\ref{mon1}, we obtain $T_{k+1}-T_{k+2}\to 0$ and deduce that ${(1-T_{k+2})}/{(k+1)}\to 0$ when $k\to \infty$. If we define $B=\sum_{j=n-1}^\infty (T_{j+1}-T_j)(1-\delta_j)$, we obtain $B=\frac{n}{n-1}T_n - T_{n-1}-\frac{1}{n-1}$.

Now, let $Q_n(t)=P_{n}(t)/(1-t)=\sum_{\ell=1}^n (1+t+\cdots + t^{\ell-1})$ which is increasing in $t$. We can write $S$ via $Q_n$ as $S= \sum_{j=n-1}^k (T_{j+1}-T_j)(1-\delta_j)Q_n(\delta_j)$. By the monotonicity of $t\mapsto Q_n(t)$, we have $S\geq Q_n(\delta_{n-1}) \cdot B$. On the other hand,
\[
S \leq (T_{n}-T_{n-1}) (1-\delta_{n-1})Q_{n}(\delta_{n-1}) + (B-(T_{n}-T_{n-1})(1-\delta_{n-1}))Q_{n}(1) < B Q_n(1),
\]
where in the first inequality we used~\ref{mon3} and in the last inequality we used $\delta_{n-1} < 1$. This shows that $Q_{n}(\delta_{n-1}) \leq S/B < Q_{n}(1)$ and from here, we can deduce that there is $\delta\in [\delta_{n-1},1 )$ such that $Q_{n}(\delta) = S/B$. Rearranging this last equality gives us
\[
S = Q_{n}(\delta) B = \frac{P_n(\delta)}{1-\delta}\left( \frac{n}{n-1}T_n - T_{n-1} -\frac{1}{n-1}  \right),
\]
which finishes the proof.
\end{proof}
}

\begin{proof}[Proof of Lemma \ref{claim:alphas}] 
We first focus on proving that there is $\varepsilon'>0$ such that for any $\varepsilon\leq \varepsilon'$,~\ref{opt-alphas} has a solution $\alpha(\varepsilon,n)=(\alpha_0(\varepsilon,n),\ldots,\alpha_{n-2}(\varepsilon,n))$ such that $0<\alpha_{0}(\varepsilon,n)< \cdots < \alpha_{n-2}(\varepsilon,n)< 1$. Furthermore, $\alpha(\varepsilon,n)$ will be a unique non-negative solution to~\ref{opt-alphas} due to the monotonicity of $P_n'$ and $\beta_n$ in $\RR_+$. From here, we will obtain immediately~\ref{alpha-i} and~\ref{alpha-ii}. Next, we show that the largest $\varepsilon'>0$ such that the above holds, denoted by $\varepsilon_n'$, holds that  $\Upsilon_{n,\varepsilon_n'}(\mathbf{y}(\varepsilon_n',n))>0$ and also $\mathbf{y}(\varepsilon,n)$ is in the interior of $K_n$ for every $\varepsilon\leq \varepsilon_n'$.

Before we go into the main proof, we prove two properties of non-negative solutions to~\ref{opt-alphas}. First, we show that $\alpha_j(\varepsilon,n) < \alpha_{j+1}(\varepsilon,n)$ for $j\in \{0,\ldots,n-3 \}$ and $\alpha_{n-2}(\varepsilon,n)< 1$. Second, we show that $\alpha_j(\varepsilon,n)$ is differentiable and decreasing in $\varepsilon$, for every $j\in \{0,\ldots,n-2\}$.

\noindent{\bf The $\alpha_j$'s are monotone in $j$.} Note that the function $P_{n}'(t)=-\sum_{\ell=1}^n \ell t^{\ell-1}$ is decreasing for $t\geq 0$, $P_{n}'(0)=-1$ and $\lim_{t\to 1}P_n'(t)=-n(n+1)/2$. 
Using the first equation in~\ref{opt-alphas}, we have
$P_n'(\alpha_{n-2}(\varepsilon,n)) = (n-1)(1+\varepsilon)- n(n+1)/2 >P_n'(1)$,
which implies that $\alpha_{n-2}(\varepsilon,n)<1$.
From the second equation in~\ref{opt-alphas}, we have
\[
P_{n}'(\alpha_{n-3}(\varepsilon,n)) = P_{n}'(\alpha_{n-2}(\varepsilon,n)) + \frac{n-2}{n-1}\left(  
 \frac{n(n+1)}{2} - \beta_n(\alpha_{n-2}(\varepsilon,n)) \right).
\]
The function $\beta_n(t)$ is strictly increasing in $(0,1)$, with $\beta_n(0)=n$ and $\lim_{t\to 1}\beta_n(t)=n(n+1)/2$. Since we have $P_{n}'(\alpha_{n-3}(\varepsilon,n))> P_{n}'(\alpha_{n-2}(\varepsilon,n))$, this implies that $\alpha_{n-3}(\varepsilon,n)< \alpha_{n-2}(\varepsilon,n)$. For $j\in \{0,1,\ldots,n-4\}$, observe that $P_n'(\alpha_{j}(\varepsilon,n))-P_n'(\alpha_{j+1}(\varepsilon,n))>0$ if and only if $\beta_n(\alpha_{j+2}(\varepsilon,n))-\beta_n(\alpha_{j+1}(\varepsilon,n))>0$.
Since $\beta_n$ is increasing, the latter holds inductively, which together with $P_n'$ being decreasing implies that $\alpha_j(\varepsilon,n)<\alpha_{j+1}(\varepsilon,n)$.
Second, we show that the $\alpha_j$ satisfying~\ref{opt-alphas} must be differentiable and monotonic in $\varepsilon$.

\noindent{\bf Differentiability and monotonicity in $\varepsilon$.} Assume that the $\alpha_j$'s satisfying~\ref{opt-alphas} exist. Now, inductively for $j=n-2,\ldots,0$, we prove that $\alpha_j$ is differentiable and decreasing. From the first equation in~\ref{opt-alphas}, we observe that $\alpha_{n-2}(\varepsilon,n)$ is differentiable. This is because $P_n'$ is invertible, strictly monotone with a differentiable inverse. Now, by deriving the same equation in $\varepsilon$, we obtain
\[
P_n''(\alpha_{n-2}(\varepsilon,n))\frac{\partial\,}{\partial \varepsilon}\alpha_{n-2}(\varepsilon,n)= n-1.
\]
Since $P_n''(\cdot)<0$, we deduce that $\partial \alpha_{n-2}(\varepsilon,n)/\partial\varepsilon <0$ and so $\alpha_{n-2}(\varepsilon,n)$ is decreasing in $\varepsilon$. We proceed similarly for $\alpha_{n-3}(\varepsilon,n)$. From the second equation in~\ref{opt-alphas}, we obtain immediately the differentiability of $\alpha_{n-3}(\varepsilon,n)$ in $\varepsilon$, and deriving the same equation gives us
\[
P_n''(\alpha_{n-3}(\varepsilon,n))\frac{\partial\,}{\partial \varepsilon}\alpha_{n-3}(\varepsilon,n) = \left( 1 + \left( \frac{n-2}{n-1} \right)\alpha_{n-2}(\varepsilon,n) \right) P_n''(\alpha_{n-2}(\varepsilon,n)) \frac{\partial \,}{\partial\varepsilon}\alpha_{n-2}(\varepsilon,n).
\]
Given that $\alpha_{n-2}(\varepsilon,n)$ is decreasing in $\varepsilon$, we have that the right-hand side of the equation is a positive term, and from where we deduce that $\alpha_{n-3}(\varepsilon,n)$ is differentiable and decreasing in $\varepsilon$. Inductively, we assume that the functions $\alpha_{j+1}(\varepsilon,n),\ldots,\alpha_{n-2}(\varepsilon,n)$ are differentiable and decreasing in $\varepsilon$. We now prove that $\alpha_{j}(\varepsilon,n)$ is differentiable and decreasing in $\varepsilon$.

Using the equations linking the $\alpha_j$'s in~\ref{opt-alphas}, we can deduce the following recursion
\begin{align}
\left(\frac{j+2}{j+1} \right) P_n'(\alpha_j) = \sum_{i=j+1}^{n-3} \frac{P_n'(\alpha_i)}{i(i+1)} + \frac{P_n'(\alpha_{n-2})}{n-2}  - \beta_n(\alpha_{j+1})+ (n-1)(1+\varepsilon) \label{eq:recursion_cvx_opt}
\end{align}
for $j=0,\ldots,n-4$. The proof of this recursion is a simple rearrangement of~\ref{opt-alphas} and we skip it for brevity. (See also Subsection~\ref{app:asymptotic-cvx} for a similar deduction in the context of the asymptotic analysis of~\ref{opt-alphas}.)

By induction, the right hand-side of Equation~\eqref{eq:recursion_cvx_opt} is differentiable in $\varepsilon$; hence, $\alpha_j$ is differentiable in $\varepsilon$. Then, by deriving~\eqref{eq:recursion_cvx_opt} in $\varepsilon$, we obtain
\begin{align*}
\left( \frac{j+2}{j+1} \right)P_n''(\alpha_j)\frac{\partial \,}{\partial\varepsilon}\alpha_j = \sum_{i=j+1}^{n-3} \frac{P_n''(\alpha_i)}{i(i+1)}\frac{\partial }{\partial \varepsilon}\alpha_i + \frac{P_n''(\alpha_{n-2})}{n-2}\frac{\partial\,}{\partial\varepsilon}\alpha_{n-2} +\alpha_{j+1}P_n''(\alpha_{j+1})\frac{\partial\, }{\partial \varepsilon}\alpha_{j+1} + n-1.
\end{align*}
Again, by induction, the right-hand side of this equation is positive. Putting everything together, we conclude that $\alpha_j$ is decreasing in $\varepsilon$. Also, from this last equation, it is possible to show that there are constants $a_{j,n}>0$, $b_{j,n}$ such that $\alpha_j(\varepsilon,n)\leq -\varepsilon \cdot a_j + b_j$ for all $j$. This, in particular, shows that $\alpha_0(\varepsilon,n)$ is $0$ at some point. (Recall that up to this point, we have assumed that $\alpha_j(\varepsilon,n)>0$, so the system has a solution.)
%\spcom{This is an aside, not sure if useful}
%Using the differential equation above, we can get $\alpha_j(\varepsilon)\leq \alpha_j(0) - \varepsilon\frac{j+1}{j+2}\frac{n-1}{n^3}$. On one side, we want the optimization problem to be nonnegative and we do $\varepsilon$ as small as possible; maybe the dual condition is to choose the largest $\varepsilon$ such that $\alpha_0(\varepsilon)>0$---as the $\alpha_j$ decrease in $\varepsilon$.
%For $\varepsilon=-1$, we have $\alpha_{n-2}(1,n)=1$. Inductively, using~\eqref{eq:recursion_cvx_opt}, we can get that $\alpha_0(1,n)=\cdots = \alpha_{n-2}(1,n)=1$.

We now move to prove the existence of a non-negative solution to~\ref{opt-alphas}. We do this by showing that for $\varepsilon=0$ there is a strictly positive solution to \ref{opt-alphas} ($\varepsilon=0$)---which in turn will be unique. Due to the differentiability of $\alpha$, we will be able to find a neighborhood $\varepsilon'>0$ where~\ref{opt-alphas} has a unique non-negative solution for $\varepsilon\leq \varepsilon'$.

\noindent{\bf For $\varepsilon=0$,~\ref{opt-alphas} has a strictly positive solution.} Fix $\varepsilon=0$. Using the first equation of~\ref{opt-alphas}, we have
$P_n'(\alpha_{n-2}(0,n)) = (n-1)-{n(n+1)}/{2} < -1 = P_n'(0).$
Since $P_n'$ is monotonic for $t\geq 0$, we obtain that there is a unique $\alpha_{n-2}(0,n)$ satisfying the equation. Furthermore, from the inequality, we obtain $\alpha_{n-2}(0,n)>0$. Now, before we proceed, note that for $t\in (0,1]$, we have $\beta_n(t) = n + \sum_{\ell=1}^n (\ell-1)t^\ell >n$. Hence, using the second equation in~\ref{opt-alphas}, we get
\begin{align*}
P_n'(\alpha_{n-3}(0,n))&=P_{n}'(\alpha_{n-2}(0,n)) + \frac{n-2}{n-1}\left( \frac{n(n+1)}{2} - \beta_n(\alpha_{n-2}(0,n))  \right) \\
& = (n-1) - \frac{n(n+1)}{2} + \left( 1 - \frac{1}{n-2} \right)\frac{n(n+1)}{2} - \frac{n-2}{n-1}\beta_n(\alpha_{n-2}(0,n)) \\
& < n-1 - \frac{n(n+1)}{2(n-2)} -\left( 1 - \frac{1}{n-2} \right) n  = - 1 - \frac{n(n+1)/2-n}{n-2} \leq -1,
\end{align*}
where in the first inequality we used that $\beta_n(\alpha_{n-2}(0,n))>n$. From here, we obtain that there is $\alpha_{n-3}(0,n)>0$ satisfying the second equation in~\ref{opt-alphas}. Assume inductively that we have found $\alpha_{j+1}(0,n),\ldots,\alpha_{n-2}(0,n)>0$. Now, using~\eqref{eq:recursion_cvx_opt}, we have
\begin{align*}
\left(  \frac{j+2}{j+1} \right)P_n'(\alpha_j(0,n))& < -\sum_{i=j+1}^{n-3}\frac{1}{i(i+1)} - \frac{1}{n-2} -n +(n-1) = - \frac{1}{j+1} - 1 = - \frac{j+2}{j+1}
\end{align*}
where in the first inequality, we used the inductive hypothesis and so $P_n'(\alpha_i(0,n))<-1$ for $i=j+1,\ldots,n-2$, and $\beta_n(\alpha_{j+1}(0,n))>n$. From here, we obtain $P_n'(\alpha_{j}(0,n))<-1$ and so there is $\alpha_j(0,n)>0$ such that the equalities in~\ref{opt-alphas} indexed by $j,j+1,\ldots$ are satisfied.
The previous results all together, imply that there is $0<\varepsilon_n'$ such that there is a unique solution $\alpha(\varepsilon,n)$ to~\ref{opt-alphas} satisfying $0<\alpha_{0}(\varepsilon,n)< \alpha_{1}(\varepsilon,n)< \cdots < \alpha_{n-2}(\varepsilon,n)<1$ for all $\varepsilon <\varepsilon_n'$, and $\alpha_0(\varepsilon_n',n)=0$.

We now prove that for any $\varepsilon\leq \varepsilon_n'$, $\mathbf{y}$ constructed from $\alpha$ as in~\ref{opt-conditions}, is an optimal solution of $\Upsilon_{n,\varepsilon}({\mathbf{y}})$ in $K_n$. Then, we use this fact to show that $\Upsilon_{n,\varepsilon_n'}(\mathbf{y}(\varepsilon_n',n))  >0$; concluding our proof.

\noindent{\bf For $\varepsilon < \varepsilon_n'$, $y(\varepsilon,n)$ is in the interior of $ K_n$.} Let $\varepsilon<\varepsilon'$. Let $\mathbf{y}(\varepsilon,n)=(y_1,\ldots,y_{n-1})$ obtained from the $\alpha_j$'s. Then, upon rearranging terms, we obtain
\begin{align*}
y_1(\varepsilon,n) & = \alpha_0(\varepsilon,n)/2,\\
y_{j+1}(\varepsilon,n) & = \alpha_j(\varepsilon,n)\left( \frac{j+1}{j+2} \right)y_j(\varepsilon,n) + \frac{1}{j(j+2)}\sum_{k=1}^j y_k(\varepsilon,n) , & \text{for } j\in \{1,\ldots,n-2\}.
\end{align*}
By applying an induction in $j$, we have $\mathbf{y}(\varepsilon,n) > 0$. Furthermore, $A_j(\mathbf{y})= \alpha_j(\varepsilon,n) y_j(\varepsilon,n) > 0$; hence, $\mathbf{y}(\varepsilon,n)\in K_n$. For notational convenience, we set $y_0=1$.

\noindent{\bf For $\varepsilon\leq \varepsilon_n'$, we have $\nabla \Upsilon_{n,\varepsilon}(y(\varepsilon,n)) = 0$.} We note that~\ref{opt-conditions} is just the first-order conditions of the problem $\min_{\mathbf{y}\geq 0} \Upsilon_{n,\varepsilon}(\mathbf{y})$ written in terms of the $x$ variables. From here, $\nabla \Upsilon_{n,\varepsilon}(\mathbf{y}(\varepsilon,n))=0$. Due to the convexity of $\Upsilon_{n,\varepsilon}$ in the interior of $K_n$ (Proposition~\ref{prop:unique-solution}), we have that $\mathbf{y}(\varepsilon,n)\in K_n$ is the unique minimizer of $\Upsilon_{n,\varepsilon}$ in $K_n$.

\noindent{\bf We have $\Upsilon_{n,\varepsilon_n'}(\mathbf{y}(\varepsilon_n',n))>0$.} Note that $\Upsilon_{n,\varepsilon}(y)$ is linear in $\varepsilon$ for any fixed $\mathbf{y}$, and so $v(\varepsilon)=\min_{y \in K_n} \Upsilon_{n,\varepsilon}(\mathbf{y})$ is concave. In particular, $v(\varepsilon)$ is continuous in $\R$. Now, let $\varepsilon^1\leq \varepsilon^2 \leq \cdots < \varepsilon_n'$ be a sequence of increasing values such that  $\varepsilon^k \to \varepsilon_n'$ as $k\to \infty$ and $\alpha_0(\varepsilon^k,n)=1/k$. The sequence $\{\varepsilon^k \}_{k\geq 1}$ exists since $\alpha_0(\cdot,n)$ is continuous. Note that $\alpha_0(\varepsilon^k,n)\to 0$ as $k\to \infty$. We claim that there is $c>0$ such that $\alpha_1(\varepsilon^k,n)\geq c$ for any $k$. By contradiction, if $\alpha_{1}(\varepsilon^k,n)\to \alpha_1(\varepsilon_n',n)=0$, then using~\ref{opt-alphas}, we can obtain that $\alpha_j(\varepsilon_n',n)=0$ for all $j$. However, this is impossible for $j=n-2$ because $P_n'(\alpha(\varepsilon_n',n))=(n-1)(1+\varepsilon_n)-n(n+1)/2<-1$ as $\varepsilon_n'<1$. Hence, $\alpha_1(\varepsilon_n',n)\geq c>0$ for some constant $c$.
The following claim states that the $y_j$'s are linearly dependent on $\alpha_0$ for $j\geq 1$. We defer its proof to Appendix~\ref{app:cvx-optimal}.

\begin{claim}\label{claim:y_linearly_dep}
For $j\geq1$, there are $\hat{y}_j(\varepsilon,n)>0$ such that $y_j(\varepsilon,n)=\alpha_0(\varepsilon,n)\hat{y}_j(\varepsilon,n)$.
\end{claim}
From the claim, we obtain immediately that $n>\hat{y}_j(\varepsilon_n',n)\geq c'>0$ for some constant $c'$. Now,
\begin{align*}
\Upsilon_{n,\varepsilon^k}(\mathbf{y}(\varepsilon^k,n))& = (1+\varepsilon^k)n\sum_{j=0}^{n-1}y_j(\varepsilon^k,n) - \sum_{j=0}^{n-2}y_j(\varepsilon^k,n)P_n(A_j(y(\varepsilon^k,n))/y_j(\varepsilon^k,n))\\
& \quad - \frac{n(n+1)}{2}\left(  \frac{n}{n-1}\sum_{j=0}^{n-1}y_j(\varepsilon^k,n) - \sum_{j=0}^{n-2}y_j(\varepsilon^k,n) - \frac{y_0(\varepsilon^k,n)}{n-1} \right) \\
& = (1+\varepsilon^k)n - P_n(2y_1(\varepsilon^k,n)) - \frac{1}{k}\sum_{j=1}^{n-2} \hat{y}_j(\varepsilon^k,n)P_n\left(A_j(y(\varepsilon^k,n))/y_j(\varepsilon^k,n)) \right)\\
& \quad - \frac{1}{k}\frac{n(n+1)}{2}\left( \frac{n}{n-1} \sum_{j=1}^{n-1} \hat{y}_j(\varepsilon^k,n) -\sum_{j=1}^{n-2} \hat{y}_j(\varepsilon^k,n)  \right)
\end{align*}
where in the second line we used that $y_0(\varepsilon,n)=1$ and the claim. Now, note that $$
\frac{A_j(\mathbf{y}(\varepsilon^k,n))}{y_j(\varepsilon^k,n)}=\frac{j+2}{j+1}\frac{\hat{y}_{j+1}(\varepsilon^k,n)}{\hat{y}_j(\varepsilon^k,n)} -\sum_{\ell=1}^j \frac{\hat{y}_\ell(\varepsilon^k,n)}{\hat{y}_j(\varepsilon^k,n)}$$
and for any $k$, this value remains non-negative and bounded. Hence, taking the limit in $k$, we obtain $v(\varepsilon^k)=\Upsilon_{n,\varepsilon^k}(\mathbf{y}(\varepsilon^k,n))\to (1+\varepsilon_n')n - P_n(0)=\varepsilon_n' n >0$. By the continuity of $v(\varepsilon)$, we obtain $v(\varepsilon_n')>0$.
\end{proof}

\begin{proof}[Proof of Lemma~\ref{lem:extension-to-infinite}]
From the definition of the sequence, we have
\begin{equation}
\mu_j(X_{\eta}(\varepsilon,n))
=\begin{cases}
\alpha_j(\varepsilon,n)(X_{j+1,\eta}(\varepsilon,n)-X_{j,\eta}(\varepsilon,n)) & \text{ for }j\in \{0,\ldots,n-2\},\\
\eta(X_{j+1,\eta}(\varepsilon,n)-X_{j,\eta}(\varepsilon,n))& \text{ for }j\ge n-1.
\end{cases}\label{eq:formula-mu}
\end{equation}
Let $\eta_0(\varepsilon,n)=\alpha_{n-2}(\varepsilon,n)$ and in what follows we consider $\eta\in [\alpha_{n-2}(\varepsilon,n),1)$; we remark that $\alpha_{n-2}(\varepsilon,n)<1$ by Lemma \ref{claim:alphas}.
We first show that the sequence $X_{\eta}(\varepsilon,n)$ satisfies condition \ref{mon1}.
Since $\lim_{j\to \infty}\mu_j(X_{\eta}(\varepsilon,n))/(X_{j+1,\eta}(\varepsilon,n)-X_{j,\eta}(\varepsilon,n))=\eta\in (0,1)$, by Claim \ref{claim:limit-t} there exists a finite value $\xi$ such that $\lim_{j\to \infty}X_{j,\eta}(\varepsilon,n)=\xi$.
In particular, this implies that $\lim_{j\to \infty}X_{j,\eta}(\varepsilon,n)/j=0$ and $\lim_{j\to \infty}(X_{j+1,\eta}(\varepsilon,n)-X_{j,\eta}(\varepsilon,n))=0$.

We proceed by induction to show that the sequence $X_{\eta}(\varepsilon,n)$ is strictly increasing.
By construction $X_{1,\eta}(\varepsilon,n)=1>X_{0,\eta}(\varepsilon,n)$, and suppose the sequence is strictly increasing up to $\ell+1$.
Together with the equality in \eqref{eq:formula-mu}, the fact that $X_{\ell+1,\eta}(\varepsilon,n)-X_{\ell,\eta}(\varepsilon,n)>0$ implies that
\begin{align*}
0<\mu_{\ell}(X_{\eta}(\varepsilon,n))=\frac{\ell+2}{\ell+1}X_{\ell+2,\eta}(\varepsilon,n)-\frac{\ell+1}{\ell}X_{\ell+1,\eta}(\varepsilon,n)+\frac{X_{1,\eta}(\varepsilon,n)}{\ell(\ell+1)},
\end{align*}
and therefore, we get
\begin{align*}
X_{\ell+2,\eta}(\varepsilon,n)&>\frac{(\ell+1)^2}{\ell(\ell+2)}X_{\ell+1,\eta}(\varepsilon,n)-\frac{1}{\ell(\ell+2)}=X_{\ell+1,\eta}(\varepsilon,n)+\frac{X_{\ell+1,\eta}(\varepsilon,n)-1}{\ell(\ell+2)}>0,
\end{align*}
where the last inequality holds since the sequence is strictly increasing up to $\ell+1$ and $1=X_{1,\eta}(\varepsilon,n)$.
Therefore, \ref{mon1} is satisfied.
 
Since, by Lemma \ref{claim:alphas}, $\alpha_j(\varepsilon,n)\in (0,1)$ for every $j\in \{0,1,\ldots,n-2\}$ and $\eta\in (0,1)$, from \eqref{eq:formula-mu} we get
$\mu_j(X_{\eta}(\varepsilon,n))\le X_{j+1,\eta}(\varepsilon,n)-X_{j,\eta}(\varepsilon,n)$ for every $j$, that is, \ref{mon2} is satisfied.
Finally, also from \eqref{eq:formula-mu}, observe that $\mu_j(X_{\eta}(\varepsilon,n))/ (X_{j+1,\eta}(\varepsilon,n)-X_{j,\eta}(\varepsilon,n))$ is equal to $\alpha_j(\varepsilon,n)$ for $j\le n-2$, and is equal to $\eta$ for $j\ge n-1$.
By Lemma \ref{claim:alphas} we have $\alpha_0(\varepsilon,n)<\cdots<\alpha_{n-2}(\varepsilon,n)\le \eta$ and we conclude that $\mu_j(X_{\eta}(\varepsilon,n))/ (X_{j+1,\eta}(\varepsilon,n)-X_{j,\eta}(\varepsilon,n))$ is non-decreasing.
Therefore, \ref{mon3} is satisfied.
This concludes the proof of the lemma.
\end{proof}

\begin{comment}
\subsection{Limit Behavior of the Exact Solution}\label{sec:limit_behavior}
Using the exact solution from Section~\ref{sec:cvx-optimal}, we can form a recurrence that converges to a solution of the integro-differential system:
\begin{align*}
    h(u(x))' + \int_x^1 \frac{1}{s} h(u(s))'\, \mathrm{d}s - \left( \frac{1-e^{-u(x)}}{u(x)}\right) - u(x)h(u(x)) &= \varepsilon,\\
    u(0) & = +\infty \\
    u(1) & = 0.
\end{align*}
where $h(u)= \int_0^1 te^{-ut}\, \mathrm{d}t =  (1-e^{-u}(1+u))/u^2$ and $\varepsilon$ is the limit of $\varepsilon_n$. We provide the analysis in Appendix~\ref{app:asymptotic-cvx}. Using ODE solvers, we obtain that $(1+\varepsilon)^{-1}\approx 0.618$. 
\end{comment}

\bibliographystyle{abbrvnat}
\bibliography{references}

\appendix

\section{Proofs Deferred from Section~\ref{sec:small_thresholds}}\label{app:small_thresholds}

\begin{proof}[Proof of Proposition~\ref{prop:smooth_instances}]
    We denote by $\mathrm{val}(\pi,F,n)$ the value obtained by running policy $\pi$ in an instance with cumulative distribution $F$ and $n$ periods. We have that for any $F$ strictly increasing and infinitely differentiable, $\mathrm{val}(\pi,F,n)\geq \beta \cdot E_n(F)$. We aim to show that for any $\varepsilon>0$, there is a policy $\pi'$ (possibly different from $\pi$) that guarantees $\mathrm{val}(\pi',F,n)\geq (1-\varepsilon)\cdot\beta \cdot E_n(F)$ for every distribution $F\in \calF$.

    Let $F\in \calF$ and $X_1,\ldots,X_n$ the $n$ random variables observed sequentially and distributed according to $F$. Using standard arguments, we can assume that $F$ is infinitely differentiable (see, e.g.,~\cite{Liu2021}).\footnote{For instance, by adding a small exponential noise to each observation, we can show that the resulting distribution is infinitely differentiable and there is a small loss in the competitive ratio guarantee that can be made arbitrarily small.} 
    %Therefore, we assume that $F$ is infinitely differentiable.
    Furthermore, by rescaling the random variables, we can assume that $E_n(F)=\sum_{\ell=1}^n \EE[\max_{i\leq \ell}X_i]=1$. Now, let $\varepsilon>0$ be small enough and define
    \[
    \hat{F}(x) = \frac{1}{1+\varepsilon/n^3} F(x) + \frac{\varepsilon/n^3}{1+\varepsilon/n^3}(1-e^{-x}).
    \]
    This is the distribution of the random variable $\hat{X}_i$ that with probability $1/(1+\varepsilon/n^3)$ follows $F$ and otherwise it follows an exponential distribution with parameter $1$. Note that $\hat{F}' (x) >0$; hence $\hat{F}$ is strictly increasing. Then, $\mathrm{val}(\pi,\hat{F},n)\geq \beta \cdot E_n(\hat{F})$.
    Now,
    \begin{align*}
    &E_n(\hat{F}) \\
    &= \sum_{\ell=1}^n \EE \left[  \max\{  \hat{X}_1,\ldots,\hat{X}_\ell \} \right] \\
    & = \sum_{\ell=1}^n \int_0^\infty x \cdot \left(  \hat{F}(x) \right)^{\ell-1} \, \mathrm{d}\hat{F}(x)\\
    & = \sum_{\ell=1}^n \int_0^\infty x \left( \frac{1}{1+\varepsilon/n^3} F(x)  + \frac{\varepsilon/n^3}{1+\varepsilon/n^3} (1-e^{-x}) \right)^{\ell-1} \left( \frac{1}{1+\varepsilon/n^3} \frac{\mathrm{d}F}{\mathrm{d}x}(x)  +  \frac{\varepsilon/n^3}{1+\varepsilon/n^3} e^{-x} \right)\, \mathrm{d}x\\
    & \geq \sum_{\ell=1}^n \frac{1}{(1+\varepsilon/n^3)^\ell} \int_0^\infty x F(x)^{\ell-1}\, \mathrm{d}F(x)\\
    & \geq \left( 1 - \frac{\varepsilon}{n^2} \right) \sum_{\ell=1}^n \EE\left[ \max\{ X_1,\ldots,X_\ell\} \right]\geq \left( 1 - \frac{\varepsilon}{n^2} \right) E_n(F).
    \end{align*}
    On the other side, consider the following policy $\hat{\pi}$ running on $F$ for an instance with $n$ periods: At time $t$, with probability $1/(1+\varepsilon/n^2)$ observe value $X_t$; otherwise sample $E_t$ from an exponential distribution of parameter $1$, and let $\hat{X}_i$ be the value observed; run $\pi$ on that value. If $\pi$ doesn't accept $\hat{X}_t$, then don't accept $X_t$; otherwise, if $\hat{X}_t$ is coming from $X_t$, accept $X_t$ while if $\hat{X}_t$ is coming from $E_t$, then stop the process and don't receive any value. Note that policy $\hat{\pi}$ gets a reward only when $\pi$ accepts $\hat{X}_t$ and $\hat{X}_t$ comes from $X_i$. Hence, the value collected by $\hat{\pi}$ can be bounded as
%    \vvcom{lo que esta en rojo no lo entiendo bien}
%    \spcom{Hice unas modificaciones al parrafo}
    \begin{align*}
    \mathrm{val}(\hat{\pi},F,n) & \geq \frac{1}{1+\varepsilon/n^3} \mathrm{val}(\pi,\hat{F},n) - n^2\frac{\varepsilon/n^3}{1+\varepsilon/n^3}  \geq \beta \left( 1 - \frac{\varepsilon}{n^2} \right)^2 E_n(F) - \frac{\varepsilon}{n}\geq  \beta\left( 1 - \frac{3\varepsilon}{n}\right) E_n(F),
    \end{align*}
    where we used that $E_n(F)=1$. Since we can make $\varepsilon$ arbitrarily small, the result follows.
\end{proof}

\begin{proof}[Proof of Proposition~\ref{prop:G_nk_few_thresholds}]
Denote by ${\mathbf{n}}_{1,\ldots,s}=(n_1,\ldots,n_s)$ the prefix of the first $s$ entries of $\mathbf{n}$ and $\mathbf{q}_{1,\ldots,s}=(q_1,\ldots,q_s)$ the prefix of the first $s$ entries of $\mathbf{q}$, then if $d_s$ represents $G_{N_s,n_s}(F,\mathbf{n}_{1,\ldots,s},\mathbf{q}_{1,\ldots,s})$ and $d_0=0$, we have,
\begin{align*}
    d_{s} &= \sum_{t=0}^{n_s-1} F(z_s)^{t} \left( (1-{F}(z_s))\EE[X_t \mid X_t \geq z_s](N_s-t) + F(z_s)\EE[X_t \mid X_t < z_s] \right) + F(z_s)^{n_s}\, \mathrm{d}_{s-1} \\
    & = \sum_{t=0}^{n_s-1}  (N_s-t) (1-q_s)^t \int_{0}^{q_s} F^{-1}(1-u)\, \mathrm{d}u + \sum_{t=0}^{n_s-1} (1-q_s)^t \int_{q_s}^1 F^{-1}(1-u)\, \mathrm{d}u + (1-q_s)^{n_s}\, \mathrm{d}_{s-1}.
\end{align*}
This poses a recursion that, upon unrolling, gives the stated equation for $G_{n,k}(\mathbf{n},\mathbf{q},F)$.
\end{proof}

\begin{proof}[Proof of Proposition~\ref{prop:monotonicity_g_n}]
Note that $g_n'(v)= {(1-(1-v)^n(1+nv))}/{v^2}\geq 0$ by using that $(1-v)^n (1+nv)\leq e^{-nv} e^{nv}=1$ with equality only occuring at $v=0$.
By deriving $g_n(v)/v$ we get
        \[
        %\left( \frac{g_n}{v}\right)'(v)= 
        \frac{1}{v^3}\Big(2 - 2 (1 - v)^n - (1 + n - (1 - v)^n + n (1 - v)^n) v\Big).
        \]
        We now show that this derivative is non-positive, equivalently,
        $2 \leq 2 (1 - v)^n + (1 + n - (1 - v)^n + n (1 - v)^n) v.$
        Let $f_n(v)= 2 (1 - v)^n + (1 + n - (1 - v)^n + n (1 - v)^n) v$. We now show that $f_n(v)\geq 2$ for all $v$.
First, we have $f_n(0)=2$.
Now we show that $f_n'(v)= ((1 + n) (1 - (1 - v)^n - v + (1 - v)^n v - n (1 - v)^n v))/(1-v) \geq 0$.
        Indeed, for $v\in (0,1)$, we have $f_n'(v) \geq 0$ if and only if $1-v -(1-v)^n + v(1-v)^n - nv(1-v)^n \geq 0,$
        which in turns is equivalent to $1 \geq (1-v)^{n-1}(1+v(n-1))$,
        which holds by the standard Bernoulli's inequality.
\end{proof}

\begin{proof}[Proof of Proposition~\ref{prop:limit_of_functions_limited_thresholds}]
    We have
    \begin{align*}
        \frac{A_{\phi n, \theta n}(\alpha/n)}{n^2} & = \frac{(1-\alpha/n)^{\theta n}( 1-(\phi-\theta+1/n)\alpha  )+(\phi +1/n)\alpha -1  }{\alpha^2} \\
        & \to \frac{e^{-\alpha \theta} (1 - (\phi - \theta)\alpha) + \phi \alpha -1 }{\alpha^2} = \Bar{A}_{\phi,\theta}(\alpha).
    \end{align*}
    We also have ${g_n(\lambda/n)}/{n}  = \left( 1 - (1-\lambda/n)\left(\frac{1-(1-\lambda/n)^n}{\lambda} \right)  \right) \to \Bar{g}(\lambda).$\qedhere
\end{proof}

\begin{proof}[Proof of Proposition~\ref{prop:Formula_En_hard_dist_1_threshold}]
    Call $a=2(e^2-3)/(e^2+1)$ and $b=4/(e^2+1)$. Then,
    \begin{align*}
        E_n & = \sum_{t=1} \int_0^1 f(u) t (1-u)^{t-1} \, \mathrm{d}u \\
        & =  an \int_0^{1/n^3} \sum_{t=1}^n t (1-u)^{t-1}\, \mathrm{d}u + \frac{b}{n} \int_{1/n^3}^{1/n^3+\beta/n} \sum_{t=1}^n t(1-u)^{t-1}\, \mathrm{d}u  = a n E_{n,1} + \frac{b}{n} E_{n,2}
    \end{align*}
    First, we have $1-1/n^2\leq (1-1/n^3)^{n}\leq (1-u)^t \leq 1$ for $u\in [0,1/n^3]$, where in the first inequality we used Bernoulli's inequality. Hence,
    \[
    \left( 1 - \frac{1}{n^2} \right) \frac{n+1}{2n^2} \leq E_{n,1} \leq \frac{n+1}{2n^2}.
    \]
    and so $an E_{n,1}\to a/2$ when $n\to \infty$. On the other hand,
    \begin{align*}
        E_{n,2}& =\sum_{t=1}^n \int_{1/n^3}^{1/n^3+\beta/n} t (1-u)^{t-1}\, \mathrm{d}u = \sum_{t=1}^n \left( 1- \frac{1}{n^3} \right)^t - \left( 1- \frac{1}{n^3} - \frac{\beta}{n} \right)^t
    \end{align*}
    By using that $1-1/n^2 \leq (1-1/n^3)^t \leq 1$, we have the following bounds on $E_{n,2}$,
    \begin{align*}
        &\left( 1 - \frac{1}{n^2} \right)n  - \left( 1 - \frac{1}{n^3} - \frac{\beta}{n} \right)\left( \frac{1- (1-1/n^3-\beta/n)^n }{1/n^3+\beta/n} \right)\\
        &\leq E_{n,2} \leq n - \left( 1 - \frac{1}{n^3} - \frac{\beta}{n} \right)\left( \frac{1- (1-1/n^3-\beta/n)^n }{1/n^3+\beta/n} \right)
    \end{align*}
    From here, we get that $(b/n)E_{n,2}\to b\left( 1- (1-e^{-\beta})/\beta\right)$.
\end{proof}

\begin{proof}[Proof of Proposition~\ref{prop:UB_ALG_1_threshold}]
    Fix $n$. We use the same notation $a=(e^2-3)/(e^2+1)$ and $b=4/(e^2+1)$ as in the previous proposition. We compute first,
    \[
    \int_0^q f(u)\, \mathrm{d}u=\begin{cases}
        an q & q\in[0,1/n^3], \\
        a/n^2+ (b/n)(q-1/n^3) & q\in [1/n^3,1/n^3+\beta/n),\\
        a/n^2 + {b \beta}/{n^2} & q\in [1/n^3+\beta/n,1].
    \end{cases}
    \]
    and
    \[
    \int_q^1 f(u)\, \mathrm{d}u = \begin{cases}
        an(1/n^3-q) + b\beta/n^2 & q\in [0,1/n^3), \\
        (b/n)(1/n^3+\beta/n-q) & q \in [1/n^3,1/n^3+\beta/n),\\
        0 & q\in [1/n^3+\beta/n,1].
    \end{cases}
    \]
    We now upper-bound
    \begin{align*}
    G_{n,1} &= \sup_{q\in [0,1]}\left\{  A_{n,n}(q)\int_0^q f(u)\, \mathrm{d}u + B_{n}(q)\int_q^1 f(u)\, \mathrm{d}u   \right\}\\
    & =\max\left\{ G_{n}^{[0,1/n^3]}, G_{n}^{[1/n^3,1/n^3+\beta/n]}, G_{n}^{[1/n^3+\beta/n,1]}  \right\},   
    \end{align*}
    where $G_n^{[a,b]}=\sup_{q\in [a,b]}\{  A_{n,n}(q)\int_0^q f(u)\, \mathrm{d}u + B_{n}(q)\int_q^1 f(u)\, \mathrm{d}u\}$, and $B_{\ell}(q) = \sum_{t=0}^{\ell-1} (1-q)^{t} = \frac{1-(1-q)^\ell}{q}$ with $\ell\in [n]$. We now analyze each term in $G_{n,1}$.

    For $G_{n}^{[0,1/n^3]}$, we have
    \begin{align*}
        G_{n}^{[0,1/n^3]}& = \sup_{q\in [0,1/n^3]}\left\{ A_{n,n}(q) anq + B_{n}(q) (an(1/n^3-q) + b\beta/n^2  )   \right\}.
    \end{align*}
    Note that \[
    \frac{1}{n^2}B_n(q) = \frac{1}{n^2}\sum_{t=0}^{n-1} (1-q)^{t} \leq \frac{1}{n},
    \]
    thus,
    \[
    \sup_{q\in [0,1/n^3]}\left\{ an q A_{n,n}(q)\right\} \leq G_{n}^{[n,1/n^3]} \leq \sup_{q\in [0,1/n^3]}\left\{ an qA_{n,n}(q)\right\} + \frac{a}{n} + \frac{b\beta}{n}.
    \]
    We can use that the function $A_{n,n}(q) q =  q\sum_{t=0}^{n-1}(n-t)(1-q)^{t}$
    is increasing for $q< 1/(n+1)$ to obtain that
    \[
    \sup_{q\in [0,1/n^3]}\left\{ anq A_{n,n}(q) \right\} = \frac{a}{n^2}\sum_{t={0}}^{n-1}(n-t)\left(1-1/n^3 \right)^t \in \left[ \left( 1- \frac{1}{n^2} \right)\frac{a(n+1)}{2n}, \frac{a(n+1)}{2n}  \right].
    \]
    Hence, $G_{n}^{[0,1/n^3]}\to a/2$.
    For $G_{n}^{[1/n^3,1/n^3+\beta/n]}$, we have
    \begin{align*}
        G_{n}^{[1/n^3,1/n^3+\beta/n]} &= \sup_{q\in [1/n^3,1/n^3+\beta/n]}\left\{A_{n,n}(q)\left( \frac{a}{n^2} + \frac{b}{n}\left( q - \frac{1}{n^3} \right)\right) + B_{n}(q) \frac{b}{n} \left( \frac{1}{n^3} +\frac{\beta}{n} - q \right)    \right\}. 
    \end{align*}
    By doing the change of variable $q=\lambda/n$, we have $\lambda\in [1/n^2,1/n^2+\beta]$ and we obtain
    \begin{align*}
        \lim_{n\to \infty}G_{n,n}^{[1/n^3,1/n^3+\beta/n]}& = \max_{\lambda\in [0,\beta]}\left\{  \Bar{A}_{1,1}(\lambda)(a+\lambda b)   \right\}.
    \end{align*}
    For $G_{n}^{[1/n^3+\beta/n,1]}$, we have
    \begin{align*}
        G_{n}^{[1/n^3,1/n^3+\beta/n,1]} & = \sup_{q\in [1/n^3+\beta/n,1]}\left\{ A_{n,n}(q) \left( \frac{a}{n^2} + \frac{b\beta}{n^2} \right)   \right\}\\
        & = \frac{A_{n,n}(1/n^3+\beta/n)}{n^2}(a+b\beta) \to \Bar{A}_{1,1}(\beta)(a+b\beta)
    \end{align*}
    since the function $A_{n,n}(q)$ is decreasing.
\end{proof}

\begin{proposition}\label{prop_app:temp_function_1_threshold}
    The function $\lambda \mapsto (e^{-\lambda}+\lambda-1)(2(e^2-3)+4\lambda)/\lambda^2$ is increasing in $[0,2]$ and decreasing in $[2,+\infty)$.
\end{proposition}

\begin{proof}
    Let $d(\lambda)=(e^{-\lambda}+\lambda-1)(2(e^2-3)+4\lambda)/\lambda^2$. A simple calculation shows
    \[
    d'(\lambda)=  \frac{e^{-\lambda}}{\lambda^3}\left( 6-e^{\lambda}(6-5\lambda)+e^{\lambda+2}(2-\lambda)+\lambda-2\lambda^2-e^2(2+\lambda) \right).
    \]
    Note that the function $\bar{d}(\lambda)= 6-e^{\lambda}(6-5\lambda)+e^{\lambda+2}(2-\lambda)+\lambda-2\lambda^2-e^2(2+\lambda) $ is dominated by the term $-(e^2-5)\lambda e^{\lambda}$ for large $\lambda$, so $\bar{g}(\lambda)\to -\infty$ as $\lambda \to \infty $. In fact, this terms dominates $\bar{d}$ for $\lambda>2$. By inspection, we conclude that $d'(\lambda)>0$ in $[0,2]$ and $d'(\lambda)<0$ in $[2,+\infty)$ which concludes the proof.
\end{proof}
%    \begin{figure}[h!]
%        \centering
%        \includegraphics[width=0.5\textwidth]{manuscript/dplot_d.png}
%        \caption{Plot of $\bar{d}$ in the range $[0,2.2]$. We note that $\bar{d}$ is $0$ at $\lambda=0$ and $\lambda=2$.}
%        \label{fig:figure_d}
%    \end{figure}

\subsection{Numerical Report for Few Thresholds}\label{app:numerical_table}

The following table reports the $\alpha$ values from Lemma~\ref{lemma:lowerbound_asymptotic_equidistant_k} for $k=1,\ldots,6$.
\begin{table}[ht]
\centering
\begin{tabular}{c c c c c c c}
\hline
$k$ & $\alpha_1$ & $\alpha_2$ & $\alpha_3$ & $\alpha_4$ & $\alpha_5$ & $\alpha_6$ \\
\hline
1 & 2 & -- & -- & -- & -- & -- \\
2 & 11.5400 & 1.3923 & -- & -- & -- & -- \\
3 & 66.7117 & 5.5709 & 0.9562 & -- & -- & -- \\
4 & 1181.3438 & 15.1417 & 3.3527 & 0.7062 & -- & -- \\
5 & 8292.3294 & 50.1045 & 7.3992 & 2.3450 & 0.5550 & -- \\
6 & 54260.5210 & 351.1725 &16.3584 & 4.5884 & 1.7943 & 0.4555 \\
\hline
\end{tabular}\caption{Optimized values of $\alpha_1,\ldots,\alpha_k$ in Lemma~\ref{lemma:lowerbound_asymptotic_equidistant_k}. Decimal values have been truncated to four significant digits.}
\end{table}

\section{Proofs Deferred from Section \ref{sec:cvx-optimal}}\label{app:cvx-optimal}
\begin{proposition}\label{prop:unique-solution}
The function $\Upsilon_{n,\varepsilon}$ is convex in the interior of $K_n$.
\end{proposition}

{ \begin{proof}
Since $P_n''(t)
=
-\sum_{\ell=2}^n \ell(\ell-1)t^{\ell-2}
\leq 0$, the function $P_n$ is concave on $\mathbb{R}_+$. Therefore, its
perspective
\[
\psi(a,b)=bP_n(a/b)
\]
is jointly concave on $\{(a,b):a\geq0,\ b>0\}$; see, e.g., section 3.2.6 in~\citep{boyd2004convex}. Since $A_j$ is
linear, $\mathbf{y}\mapsto 
\psi\bigl(A_j(\mathbf{y}),y_j\bigr)$ is concave in the interior of $K_n$ for every
$j\in\{1,\ldots,n-2\}$. Moreover, $\mathbf{y}\mapsto P_n(2y_1)$ is concave because it is the composition of the concave function $P_n$ with a linear map. Finally, $\Upsilon_{n,\varepsilon}$ is an affine function minus these concave functions. Hence, $\Upsilon_{n,\varepsilon}$ is convex in the interior of $K_n$.
\end{proof}}

\begin{proof}[Proof of Claim \ref{claim:identity-expectation}]
Let $G$ be the cumulative distribution of $\max(s,X)$.
Then, $1-G(t)$ is equal to 1 if $t\le s$ and $1-F(t)$ if $t>s$.
Therefore, 
$$\EE[\max(s,X)]=s+\int_{s}^{\infty}(1-F(t))\mathrm dt=s+\EE[X]-\int_{0}^{s}(1-F(t))\mathrm dt=\EE[X]+\int_{0}^{s}F(t)\mathrm dt.\qedhere$$
\end{proof}

\begin{proof}[Proof of Claim \ref{claim:tau-conv}]
We have $G_j(F)/j^2\le \sum_{i=1}^j\EE[\max\{X_1,\ldots,X_i\}]/j^2\le \EE[\max\{X_1,\ldots,X_j\}]/j$, and $\EE[\max\{X_1,\ldots,X_j\}]=o(j)$ (see, e.g., \cite{correa2021asymptotic}, \cite{downey1990distribution}).
We conclude that $\lim_{j\to \infty}\tau_j(F)/j= 0$.
Observe that $\lim_{j\to \infty}(\tau_{j+1}(F)-\tau_{j}(F))=0$ directly when there exists a finite value $\tau$ such that $\tau_j(F)\to \tau$.
Otherwise, suppose that $G_j(F)/j=\tau_j(F)\to \infty$.
From the recurrence satisfied by the optimal policy, we have that $(1+1/j)\tau_{j+1}(F)-\tau_j(F)$ is equal to
\begin{align*}
\frac{G_{j+1}(F)-G_j(F)}{j}=\frac{1}{j}+\EE[\max(0,X-G_j(F)/j)]=\frac{1}{j}+\int_{\tau_j(F)}^{\infty}(1-F(s))\mathrm ds,
\end{align*}
where $X$ is distributed according to $F$.
Since the expectation of $F$ is finite, and $\tau_j(F)\to \infty$, we have that $\int_{\tau_j(F)}^{\infty}(1-F(s))\mathrm ds$ converges to zero as $j\to \infty$, and therefore we conclude that $(1+1/j)\tau_{j+1}(F)-\tau_j(F)\to 0$.
Since $\tau_{j+1}(F)/j=(1+1/j)\tau_{j+1}(F)/(j+1)\to 0$ by the first part of the claim, and $(1+1/j)\tau_{j+1}(F)-\tau_j(F)=\tau_{j+1}(F)/j+\tau_{j+1}(F)-\tau_j(F)$, we conclude that $\tau_{j+1}(F)-\tau_j(F)\to 0$.
\end{proof}

\begin{proof}[Proof of Claim~\ref{claim:limitoftau}]
Now we prove that $\lim_{j\to \infty}\tau_j(F)=\omega_1(F)$. 
Suppose, for the sake of contradiction, that $\lim_{j\to \infty}\tau_j(F)=\calT<\omega_1(F)$, with $\calT$ finite, and let $M$ be the minimum between $\calT+0.5$, and $(\calT+\omega_1(F))/2$.
We remark that as $\omega_1(F)$ could be $\infty$, this definition of $M$ guarantees that $M$ is finite and strictly less than $\omega_1(F)$ in any case.
Observe that since $\tau_j(F)$ is non-decreasing, it holds that $\tau_j(F)\le \calT$ for every $j$.
Then, from the identity \eqref{recur:ident}, for every $j\ge 1$ we get
\begin{align}
&\frac{(j+1)\tau_{j+1}(F)-\tau_1(F)}{j}-\tau_j(F)\notag\\
&=\EE[\max(0,X-\tau_j(F))]
%&=\int_{\tau_{j}(F)}^{\infty}(1-F(x))\mathrm dx\notag\\
= \int_{\tau_{j}(F)}^{\omega_1(F)}(1-F(x))\mathrm dx\ge \int_{\calT}^{M}(1-F(x))\mathrm dx\ge (M-\calT)(1-F(M)),\label{recur:ineq0}
\end{align}
where the third equality holds since $F(x)=1$ for $x>\omega_1(F)$, the first inequality holds since $\tau_j(F)\le \calT<M<\omega_1(F)$, and the last inequality follows from $1-F$ being non-increasing.
Observe that $(M-\calT)(1-F(M))$ is strictly positive, since $F(M)<1$, as $M<\omega_1(F)$.
Then, for every $k,\ell\ge 1$, with $k\ge \ell+1$, we have 
\begin{align}
\tau_{k}(F)-\tau_{\ell}(F)
%&=\sum_{j=\ell}^{k-1}(\tau_{j+1}(F)-\tau_j(F))\notag\\
&=\sum_{j=\ell}^{k-1}\left(\frac{j+1}{j}\tau_{j+1}(F)-\frac{\tau_1(F)}{j}-\tau_j(F)\right)-\sum_{j=\ell}^{k-1}\frac{\tau_{j+1}(F)-\tau_1(F)}{j}\notag\\
&\ge (k-\ell)(M-\calT)(1-F(M))-\sum_{j=\ell}^{k-1}\frac{\tau_{j+1}(F)-\tau_1(F)}{j},\label{recur:ineq1}
\end{align}
where the inequality holds from \eqref{recur:ineq0}.
From Claim \ref{claim:tau-conv}, there exists an integer $j_0$ such that $(\tau_{j+1}(F)-\tau_1(F))/j<(M-\calT)(1-F(M))/2$ for every $j\ge j_0$.
Together with \eqref{recur:ineq1}, we have
\begin{align*}
\tau_{k}(F)-\tau_{\ell}(F)&\ge (k-\ell)(M-\calT)(1-F(M))-\frac{1}{2}(k-\ell)(M-\calT)(1-F(M))\\
&=\frac{1}{2}(k-\ell)(M-\calT)(1-F(M))\to \infty \text{ as }k\to \infty,
\end{align*}
which contradicts that $\tau_k(F)\le \calT<M$.
\end{proof}

\begin{proof}[Proof of Claim \ref{claim:limit-t}]
Since the sequence $T$ is strictly increasing, proving that the sequence is upper-bounded is sufficient to show the convergence.
Since $\mu_j(\mathbf{T})\le \calL(T_{j+1}-T_j)$ for every $j\ge 0$, we have for every $k\ge 2$
\begin{align*}
\calL\sum_{j=1}^{k-1}(T_{j+1}-T_j)&\ge \sum_{j=1}^{k-1}\left(\frac{j+2}{j+1}T_{j+2}-\frac{j+1}{j}T_{j+1}+\frac{T_1}{j(j+1)}\right)=\frac{k+1}{k}T_{k+1}-2T_2+T_1-\frac{T_1}{k},
\end{align*}
and therefore, by expanding the telescopic sum on the left-hand side, we get 
$$\calL(T_k-T_1)\ge \frac{k+1}{k}T_{k+1}-2T_2+T_1-\frac{T_1}{k}.$$
Since $2T_2-2T_1=\mu_0(T)\le \calL T_1$, for every $k\ge 2$ we have
$\calL T_k\ge (1+1/k)T_{k+1}-T_1(1+1/k),$
which further implies that $2T_1+\calL T_k\ge T_{k+1}$ for every $k\ge 2$.
Consider $M=2T_1/(1-\calL)+0.1$.
We are done if $T_j\le M$ for every $j$, as the sequence is therefore upper-bounded.
Otherwise, let $k(M)$ be the last time that the strictly increasing sequence $T$ is below $M$, namely, $T_{k(M)+1}\ge M$ and $T_{k(M)}<M$.
Then, since $2T_1+\calL T_{k(M)}\ge T_{k(M)+1}$, we get $2T_1+\calL M>M$, that is, $M<2T_1/(1-\calL)$, which is a contradiction.
We conclude that $T_j\le M$ for every $j$; therefore, the sequence is bounded.
\end{proof}

\begin{proof}[Proof of Claim \ref{claim:1and2}]
Just by expanding the expectation using the definition of $H$, we get
\begin{align}
G_{1}(H)&=\sum_{\ell=0}^{\infty}\int_{T_{\ell}}^{T_{\ell+1}}(1-H(x))\mathrm dx=\sum_{\ell=0}^{\infty}\int_{T_{\ell}}^{T_{\ell+1}}\left(1-\frac{\mu_{\ell}(T)}{T_{\ell+1}-T_{\ell}}\right)\mathrm dx=\sum_{\ell=0}^{\infty}(T_{\ell+1}-T_{\ell}-\mu_{\ell}(T)),\label{lem:cvx-lem2-eq1}
\end{align}
On the other hand, by expanding $\mu_{\ell}(T)$, we get that for every $k$ the following holds:
\begin{align}
\sum_{\ell=1}^{k}(T_{\ell+1}-T_{\ell}-\mu_{\ell}(T))&=\sum_{\ell=1}^{k}\left(T_{\ell+1}-T_{\ell}-\left(\frac{\ell+2}{\ell+1}T_{\ell+2}-\frac{\ell+1}{\ell}T_{\ell+1}+\frac{T_1}{\ell(\ell+1)}\right)\right)\notag\\
%&=\sum_{\ell=1}^{k}\left(T_{\ell+1}-T_{\ell}-\left(\frac{\ell+2}{\ell+1}T_{\ell+2}-\frac{\ell+1}{\ell}T_{\ell+1}\right)-\left(\frac{T_1}{\ell}-\frac{T_1}{\ell+1}\right)\right)\notag\\
&=T_{k+1}-T_1+2T_{2}-\frac{k+2}{k+1}T_{k+2}-T_1+\frac{T_1}{k+1}\label{lem:cvx-lem2-eq2},
\end{align}
and therefore, from \eqref{lem:cvx-lem2-eq1} and \eqref{lem:cvx-lem2-eq2} we get
\begin{align*}
G_1(H)&=\lim_{k\to \infty}\left(T_1-\mu_0(T)+T_{k+1}-T_1+2T_{2}-\frac{k+2}{k+1}T_{k+2}-T_1+\frac{T_1}{k+1}\right)\\
&=\lim_{k\to \infty}\left(T_1-(2T_2-2T_1)+T_{k+1}-T_1+2T_{2}-\frac{k+2}{k+1}T_{k+2}-T_1+\frac{T_1}{k+1}\right)\\
&=\lim_{k\to \infty}\left(T_1+T_{k+1}-\frac{k+2}{k+1}T_{k+2}+\frac{T_1}{k+1}\right)=\lim_{k\to \infty}\left(T_1+T_{k+1}-T_{k+2}+\frac{T_{k+2}}{k+1}+\frac{T_1}{k+1}\right)=T_1,
\end{align*}
where the last limit holds from conditions \ref{mon1}-\ref{mon2}.
This finishes the proof.
\end{proof}

\begin{proof}[Proof of Claim \ref{claim:tailsum}]
For every $j\ge n-1$ we have
$\mu_j(X_{\eta}(\varepsilon,n))=\eta(X_{j+1,\eta}(\varepsilon,n)-X_{j,\eta}(\varepsilon,n)).$
If we let $\calS^{\star}=\sum_{j=n-1}^{\infty}(X_{j+1,\eta}(\varepsilon,n)-X_{j,\eta}(\varepsilon,n))$, we get
\begin{align*}
\eta \calS^{\star}&=\sum_{j=n-1}^{\infty}(X_{j+2,\eta}(\varepsilon,n)-X_{j+1,\eta}(\varepsilon,n))+\sum_{j=n-1}^{\infty}\left(\frac{1}{j+1}X_{j+2,\eta}(\varepsilon,n)-\frac{1}{j}X_{j+1,\eta}(\varepsilon,n)+\frac{1}{j(j+1)}\right)\\
%&=\calS^{\star}-(X_{n,\eta}(\varepsilon,n)-X_{n-1,\eta}(\varepsilon,n))+\sum_{j=n-1}^{\infty}\left(\frac{1}{j+1}(X_{j+2,\eta}(\varepsilon,n)-1)-\frac{1}{j}(X_{j+1,\eta}(\varepsilon,n)-1)\right)\\
&=\calS^{\star}-(X_{n,\eta}(\varepsilon,n)-X_{n-1,\eta}(\varepsilon,n))-\frac{1}{n-1}(X_{n,\eta}(\varepsilon,n)-1)\\
%&=\calS^{\star}-\left(\frac{x_{n}(\varepsilon,n)}{n}-\frac{x_{n-1}(\varepsilon,n)}{n-1}\right)-\frac{1}{n-1}\left(\frac{x_{n}(\varepsilon,n)}{n}-x_1(\varepsilon,n)\right)\\
&=\calS^{\star}-\left(\frac{x_{n}(\varepsilon,n)}{n-1}-\frac{x_{n-1}(\varepsilon,n)}{n-1}-\frac{x_{1}(\varepsilon,n)}{n-1}\right),
\end{align*}
and therefore $\calS^{\star}P_n(\eta)=\frac{P_n(\eta)}{1-\eta}\cdot \frac{x_n(\varepsilon,n)-x_{n-1}(\varepsilon,n)-x_1(\varepsilon,n)}{n-1}$.
\end{proof}

\begin{proof}[Proof of Claim \ref{claim:mon-x}]
Consider a value of $\eta\in (\eta_0(\varepsilon,n),1)$.
By Lemma \ref{lem:extension-to-infinite}, the sequence $X_{\eta}(\varepsilon,n)$ satisfies the conditions \ref{mon1}-\ref{mon3}.
Therefore, by Lemma \ref{lem:sequence-to-dist}, there exists a distribution $H$ such that $X_{j,\eta}(\varepsilon,n)=G_j(H)/j$ for every $j$.
In particular, this implies that $x_j(\varepsilon,n)=G_j(H)$ for every $j\le n$, and therefore $G_1(H)=x_1(\varepsilon,n)=1$.
On the other hand, from the recursion satisfied by the optimal policy, for every $j\in \{0,1,\ldots,n-1\}$, we have $x_{j+1}(\varepsilon,n)=G_{j+1}(H)=G_1(H)+\EE[\max(G_j(H),jZ)]\ge G_1(H)+G_j(H)=x_{1}(\varepsilon,n)+x_{j}(\varepsilon,n)$, with $Z$ distributed according to $H$.
This concludes the proof.
\end{proof}

\begin{proof}[Proof of Claim~\ref{claim:y_linearly_dep}]
For $j=1$, the result is immediately true with $\hat{y}_1(\varepsilon,n)=1/2$. Assume that the result is also true for $1,\ldots,j$ for some $j\geq 1$. Then,
\begin{align*}
y_{j+1}(\varepsilon,n) &= \alpha_j(\varepsilon,n)\left( \frac{j+1}{j+2} \right)y_j + \frac{1}{j(j+2)}\sum_{k=1}^j y_k\\
&= \alpha_j(\varepsilon,n)\left( \frac{j+1}{j+2} \right)\alpha_0(\varepsilon,n)\hat{y}_j(\varepsilon,n) + \frac{1}{j(j+2)}\sum_{k=1}^j \alpha_0(\varepsilon,n)\hat{y}_k(\varepsilon,n)\\
& = \alpha_0(\varepsilon,n) \left( \alpha_j(\varepsilon,n)\left( \frac{j+1}{j+2} \right)\hat{y}_j(\varepsilon,n) + \frac{1}{j(j+2)}\sum_{k=1}^j \hat{y}_k(\varepsilon,n) \right)
\end{align*}
where in the second line we use the inductive hypothesis. From the last line, the proof follows.
\end{proof}

\end{document}